\documentclass{llncs}

\usepackage{amssymb,amsmath,amscd,latexsym}
\usepackage[vlined,ruled]{algorithm2e}
\usepackage{pstricks}
\usepackage{graphicx}
\usepackage{subfigure}
\usepackage{macro}

\usepackage{tikz}
\usetikzlibrary{arrows,shapes}
\tikzstyle{ha}=[ellipse, draw=black, semithick]
\tikzset{my_rect1/.style={fill=black!30, thin}}
\tikzset{my_rect2/.style={fill=black!15, thin}}
\tikzset{my_rect3/.style={fill=black!5, thin}}
\tikzset{cprefig/.style={scale=0.7}}

\begin{document}
\title{Automatic Synthesis of Switching Controllers for 
Linear Hybrid Automata\\[2ex]
{\rm Technical Report}}

\author{Massimo Benerecetti \and Marco Faella \and Stefano Minopoli}
\institute{Universit\`a di Napoli\\``Federico II'', Italy}
\maketitle

\begin{abstract}
  In this paper we study the problem of automatically generating
  switching controllers for the class of Linear Hybrid Automata, with
  respect to safety objectives. 
  We identify and solve inaccuracies contained in previous
  characterizations of the problem, providing a sound and
  complete symbolic fixpoint procedure,
  based on polyhedral abstractions of the state space.
  We also prove the termination of each iteration of the procedure.
  %, thus  confirming semi-decidability of the synthesis problem.  
  Some promising experimental results are presented, based on an
  implementation of the fixpoint procedure on top of the tool PHAVer.
\end{abstract}

\section{Introduction}

Hybrid systems are an established formalism for modeling
physical systems which interact with a digital controller.
From an abstract point of view, a hybrid system is a dynamic system
whose state variables are partitioned into discrete and continuous ones.
Typically, continuous variables represent physical quantities like
temperature, speed, etc., while discrete ones represent \emph{control modes},
i.e., states of the controller.

Hybrid automata~\cite{hybridLics96} are the most common syntactic variety of hybrid system:
a finite set of locations, similar to the states of a finite automaton,
represents the value of the discrete variables.
The current location, together with the current value of 
the (continuous) variables, form the instantaneous description of the system.
Change of location happens via discrete transitions, and 
the evolution of the variables is governed 
by differential equations attached to each location.
In a Linear Hybrid Automaton (LHA), the allowed differential
equations are in fact differential inclusions of the type $\dot{\mathbf{x}} \in P$,
where $\dot{\mathbf{x}}$ is the vector of the first derivatives of all variables
and $P$ is a convex polyhedron.
Notice that differential inclusions are non-deterministic, allowing
for infinitely many solutions.

The most studied problem for hybrid systems is \emph{reachability}:
computing the set of states that are reachable from the initial states,
in any amount of time.
The reachability problem for LHAs was proved undecidable in~\cite{whatsdecidable95}, indicating that no exact discrete abstraction exists.
The complexity standing of the problem was further refined to 
\emph{semi-decidable} in~\cite{WongToi}, whose results imply that
it is possible to exactly compute the set of states that are reachable 
within a bounded number of discrete transitions 
(\emph{bounded-horizon reachability}).

We study LHAs whose discrete transitions are partitioned into
controllable and uncontrollable ones, and we wish to compute a
strategy for the controller to satisfy a given goal, regardless of the
evolution of the continuous variables and of the uncontrollable
transitions. Hence, the problem can be viewed as a \emph{two player
  game}~\cite{tomlin00}: 
on one side the controller, who can only issue controllable
transitions, on the other side the environment, who can choose the
trajectory of the variables and can take uncontrollable transitions whenever
they are enabled.

As control goal, we consider safety, i.e., the objective of keeping
the system within a given region of safe states.  This problem has
been considered several times in the literature. Here, we fix some
inaccuracies in previous presentations, propose a sound and
complete procedure for the problem\footnote{In other words, an algorithm that may or may not
terminate, and that provides the correct answer whenever it terminates.}
and we present a publicly available implementation of the procedure.
In particular, we present
a novel algorithm for computing the set of states that may reach a
given region while avoiding another one, a problem that is at the
heart of the synthesis procedure.

Contrary to most recent literature on the subject, we focus on exact 
algorithms. Although it is established that exact analysis and synthesis
of realistic hybrid systems is computationally demanding, %prohibitive?
we believe that the ongoing research effort on approximate techniques
should be based on the solid grounds provided by the exact approach.
For instance, a tool implementing an exact algorithm (like our PHAVer+)
may serve as a benchmark to evaluate the performance and the precision
of an approximate tool.

%The problem of automatic synthesis for reachability objectives is entirely novel.
%Although the two control goal we examine are dual, the corresponding synthesis
%problems are not, because the game is asymmetric, 
%in that the continuous behavior is always uncontrollable.

\paragraph*{Related work.}
The idea of automatically synthesizing controllers for dynamic systems
arose in connection with discrete systems~\cite{RW87}.
%similar to finite automata~\cite{RW87}.
Then, the same idea was applied to real-time systems modeled by timed
automata~\cite{MalerGames}, thus coming one step closer to the
continuous systems that control theory usually deals with.  Finally,
it was the turn of hybrid systems~\cite{WongToi,hybridgames99}, and in
particular of Linear Hybrid Automata, the very model that we analyze
in this paper.  Wong-Toi proposed the first symbolic semi-procedure to
compute the controllable region of a LHA w.r.t. a safety
goal~\cite{WongToi}.  The heart of the procedure lies in the operator
\emph{flow\_avoid}$(U,V)$, which computes the set of system
configurations from which a continuous trajectory may reach the set
$U$ while avoiding the set $V$ (hence, in this paper we call this
operator $\rwa$, for \emph{Reach While Avoiding}).
Tomlin et al.~\cite{tomlin00} and Balluchi et al.~\cite{Balluchi03}
analyze much more expressive models, with generality in mind rather than
automatic synthesis. Their \emph{Reach} and
\emph{Unavoid\_Pre} operators, respectively, again correspond to
\emph{flow\_avoid}.

As explained in Section~\ref{sec:previous}, the algorithm provided
in~\cite{WongToi} for \emph{flow\_avoid} does not work for non-convex
$V$, a case which is very likely to occur in practice, even if the
original safety goal is convex.
A slightly different algorithm for \emph{flow\_avoid} is reported 
to have been implemented in the tool {\sc HoneyTech}~\cite{honeytech},
and we compare it with ours in Section~\ref{sec:previous}.

Asarin et al.~\cite{asarin00}
investigate the synthesis problem for hybrid systems
where all discrete transitions are controllable and the trajectories
satisfy given linear differential equations of the type
$\dot{\mathbf{x}} = A \mathbf{x}$.  The expressive
power of these constraints is incomparable with the one offered by the
differential inclusions occurring in LHAs.  In particular, linear
differential equations give
rise to deterministic trajectories, while differential inclusions are
non-deterministic.  In control theory terms, differential inclusions
can represent the presence of environmental \emph{disturbances}.  The
tool {\tt d/dt}~\cite{ddt02}, by the same authors, is reported to
support controller synthesis for safety objectives, but the publicly
available version in fact does not.

The rest of the paper is organized as follows.  Section~\ref{sec:defs}
introduces and motivates the model.  In Section~\ref{sec:safety}, we
present the semi-procedure which solves the synthesis problem.
Section~\ref{sec:experiments} reports some experiments performed on
our implementation of the procedure,
while Section~\ref{sec:conclusions} draws some conclusions.
% Concluding remarks and future
%work are given in Section~\ref{sec:conclusion}.

\section{Linear Hybrid Automata} \label{sec:defs}

A \emph{convex polyhedron} is a subset of $\reals^n$ that is the intersection
of a finite number of strict and non-strict affine half-spaces.
A \emph{polyhedron} is a subset of $\reals^n$ that is 
the union of a finite number of convex polyhedra.
For a general (i.e., not necessarily convex) polyhedron $G \subseteq \reals^n$, 
we denote by $\clos{G}$ its topological closure,
and by $\conv{G} \subseteq 2^{\reals^n}$ its representation
as a finite set of convex polyhedra.
%
%In a \emph{Linear Hybrid Automaton (LHA)}, invariants
%and initial states are given by polyhedra over $X$, flow predicates by convex
%polyhedra over $\dot{X}$, and jump relations by polyhedra over $X \cup X'$.

Given an ordered set $X = \{x_1, \ldots , x_n\}$ of variables, a \emph{valuation} 
is a function $v : X\rightarrow\reals$.
Let $\Val(X)$ denote the set of valuations over $X$.
There is an obvious bijection between $\Val(X)$ and $\reals^n$,
allowing us to extend the notion of (convex) polyhedron to sets of valuations.
We denote by $\cpoly(X)$ (resp., $\poly(X)$) the set of convex polyhedra 
(resp., polyhedra) on $X$.

%When a valuation $u$ over $X$ and a valuation $v$ over $Y$ agree on the shared
%variables, i.e., $u(x) = v(x)$ for all $x \in X \cap Y$, 
%we use $u\sqcup v$ to denote the \emph{joint} valuation $w$ over $X \cup Y$
%defined by $w(x) = u(x)$ if $x \in X$ and $w(x) = v(x)$ otherwise.
We use $\dot{X}$ to denote the set $\{\dot{x}_1, \ldots , \dot{x}_n\}$
of dotted variables, used to represent the first derivatives, and $X'$
to denote the set $\{x'_1, \ldots , x'_n\}$ of primed variables, used
to represent the new values of variables after a transition.
Arithmetic operations on valuations are defined in the straightforward way.
An \emph{activity} over $X$ is a differentiable function 
$f:\reals^{\geq 0}\rightarrow \Val(X)$.
%such that its first derivative is Lipschitz continuous.
% The role of the Lipschitz continuity assumption
%will become clear in Section~\ref{sec:reach}.
%\mynote{Serve Lipschitz continuity?}
Let $Acts(X)$ denote the set of activities over $X$.
The \emph{derivative} $\dot{f}$ of an activity $f$ is defined in the standard way and
it is an activity over $\dot{X}$.
%defined analogously to the derivative in $\reals^n$.
%The extension of operators from valuations to activities is done pointwise.
%Let $const_X(Y) = \{(v, v') \vert v, v'\in V (X), v\proj_Y = v'\proj_Y\}$.
A \emph{Linear Hybrid Automaton} $H = (\Loc, X, \edgc, \edgu, \Flow,
\Inv, \Init)$ consists of the following:
\begin{itemize}
\item A finite set $\Loc$ of \emph{locations}.
\item A finite set $X = \{x_1, \ldots, x_n\}$ of continuous,
  real-valued \emph{variables}.  A \emph{state} is a pair $\pair{l, v}$ of
  a location $l$ and a valuation $v \in \Val(X)$.
\item Two sets $\edgc$ and $\edgu$ of \emph{controllable} and
  \emph{uncontrollable transitions}, respectively. They describe
  instantaneous changes of locations, in the course of which variables
  may change their value.  Each transition $(l, \mu, l')\in \edgc \cup
  \edgu$ consists of a \emph{source location} $l$, a \emph{target
    location} $l'$, and a \emph{jump relation} $\mu \in \poly(X \cup
  X')$, that specifies how the variables may change their value during
  the transition.  The projection of $\mu$ on $X$ describes the
  valuations for which the transition is enabled; this is often
  referred to as a \emph{guard}.
\item A mapping $\Flow : \Loc \to \cpoly(\dot{X})$ attributes to each
  location a set of valuations over the first derivatives of the
  variables, which determines how variables can change over time.
\item A mapping $\Inv : \Loc \to \poly(X)$, called the
  \emph{invariant}.
\item A mapping $\Init: \Loc \to \poly(X)$, contained in the
  invariant, defining the \emph{initial states} of the automaton.
\end{itemize}
We use the abbreviations $S = \Loc \times \Val(X)$ for the set of states
and $\edg = \edgc \cup \edgu$ for the set of all transitions.
Moreover, we let $\Invs = \bigcup_{l\in\Loc} \{ l \} \times \Inv(l)$
and $\Inits = \bigcup_{l\in\Loc} \{ l \} \times \Init(l)$.
Notice that $\Invs$ and $\Inits$ are sets of states.
Given a set of states $A$ and a location $l$, we
denote by $\proj{A}{l}$ the projection of $A$ on $l$, i.e. 
$\{v\in \Val(X) \mid \pair{l,v}\in A\}$.

\subsection{Semantics}

The behavior of a LHA is based on two types
of transitions: \emph{discrete} transitions correspond
to the $\edg$ component, and produce an instantaneous change
in both the location and the variable valuation;
\emph{timed} transitions describe the change of the variables over
time in accordance with the $\Flow$ component.

Given a state $s = \pair{l, v}$, we set $\loc(s) = l$ and $\val(s) = v$.
An activity $f \in Acts(X)$ is called \emph{admissible from $s$}
if \emph{(i)} $f(0) = v$ and
\emph{(ii)} for all $\delta\geq 0$ it holds 
$\dot{f}(\delta) \in \Flow(l)$.
We denote by $\adm(s)$ the set of activities that are admissible from $s$.
Additionally, for $f \in \adm(s)$, the \emph{span} of $f$ in $l$, denoted by $\sp(f,l)$
is the set of all values $\delta\geq 0$
such that $\pair{l, f(\delta')} \in \Invs$ for all $0\leq\delta'\leq\delta$.
Intuitively, $\delta$ is in the span of $f$ iff $f$ never leaves the invariant
in the first $\delta$ time units.
If all non-negative reals belong to $\sp(f,l)$, we write $\infty \in \sp(f,l)$.

\paragraph{Runs.}
Given two states $s, s'$, and a transition $e \in \edg$,
there is a \emph{discrete step} $s \xto{e} s'$ 
with \emph{source} $s$ and \emph{target} $s'$ iff
\emph{(i)} $s, s' \in \Invs$, 
\emph{(ii)} $e = (loc(s), \mu, loc(s'))$, and
\emph{(iii)} $(\val(s), \val(s')[X'/X])\in \mu$,
where $\val(s')[X'/X]$ is the valuation in $\Val(X')$ obtained from $s'$ by renaming each variable in $X$ with the corresponding primed variable in $X'$.
There is a \emph{timed step} $s \xto{\delta,f} s'$ with \emph{duration} 
$\delta\in\reals^{\geq 0}$ and activity $f \in \adm(s)$ iff 
\emph{(i)} $s \in \Invs$, 
\emph{(ii)} $\delta \in \sp(f,\loc(s))$, and
\emph{(iii)} $s' = \pair{\loc(s), f(\delta)}$.
For technical convenience, we admit timed steps of duration zero%
\footnote{Timed steps of duration zero can be disabled by
adding a clock variable $t$ to the automaton and requesting that each
discrete transition happens when $t>0$ and resets $t$ to $0$ when taken.}.
A special timed step is denoted $s \xto{\infty,f}$
and represents the case when the system follows an activity forever.
This is only allowed if $\infty \in \sp(f,\loc(s))$.
Finally, a \emph{joint step} $s \xto{\delta,f,e} s'$ 
represents the timed step $s \xto{\delta,f} \pair{\loc(s),f(\delta)}$
followed by the discrete step $\pair{\loc(s),f(\delta)} \xto{e} s'$.

A \emph{run} is a sequence
\begin{equation} \label{eq:run} 
r = s_0 \xto{\delta_0,f_0} s_0' \xto{e_0}
      s_1 \xto{\delta_1,f_1} s_1' \xto{e_1} s_2 \ldots s_n \ldots
%$r = s_0\xto{\delta_0,f_0,a_0} s_1
%        \xto{\delta_1,f_1,a_1} \ldots \xto{\delta_{n-1},f_{n-1},a_{n-1}} s_n \ldots$
\end{equation}
of alternating timed and discrete transitions, 
such that either the sequence is infinite, or it ends with a timed transition 
of the type $s_n \xto{\infty,f}$.  
If the run $r$ is finite, we define $\length(r)=n$ to be the length of the run, 
otherwise we set $\length(r)=\infty$.  
The above run is \emph{non-Zeno} if for all $\delta\geq 0$ there exists $i\geq 0$ 
such that $\sum_{j=0}^{i} \delta_j > \delta$.  
We denote by $\states(r)$ the set of all states visited by $r$.  
Formally, $\states(r)$ is the smallest set containing
all states $\pair{\loc(s_{i}), f_{i}(\delta)}$,
for all $0 \leq i \leq \length(r)$ and all 
$0 \leq \delta \leq \delta_i$.
%
%For a non-Zeno run $r$ and a time $\delta \geq 0$, with an abuse of notation we denote
%by $r(\delta)$ the state encountered by the run at time $\delta$.
%Formally, let $i$ be such that 
%$\sum_{j=0}^{i-1} \delta_j \leq \delta < \sum_{j=0}^{i} \delta_j$,
%we have $r(\delta) = \pair{\loc(s_{i}), f_{i}(\delta - \sum_{j=0}^{i-1} \delta_j)}$. The function $r: \reals^{\geq 0} \to S$ so defined can be called
%the \emph{flat view} of the run $r$.
%
Notice that the states from which discrete transitions start
(states $s_i'$ in~\eqref{eq:run}) appear in $\states(r)$.
Moreover, if $r$ contains a sequence of one or more
zero-time timed transitions, all intervening states appear in $\states(r)$.

\paragraph{Zenoness and well-formedness.}
A well-known problem of real-time and hybrid systems is that definitions 
like the above admit runs that take infinitely many discrete transitions 
in a finite amount of time (i.e., \emph{Zeno} runs), 
even if such behaviors are physically meaningless.  
In this paper, we assume that the hybrid automaton under consideration
generates no such runs. This is easily achieved by using 
an extra variable, representing a clock, to ensure that the delay between 
any two transitions is bounded from below by a constant.
%adding an extra variable $t$,
%representing a clock, to the model, and ensuring that all transitions
%can only be taken when $t\geq 1$ and that they reset $t$ to zero.
We leave it to future work to combine our results with more
sophisticated approaches to Zenoness known in the
literature~\cite{Balluchi03,dAFHMS03}.

Moreover, we assume that the hybrid automaton under consideration
is \emph{non-blocking}, i.e., whenever the automaton is about
to leave the invariant there must be an uncontrollable transition
enabled. Formally, for all states $s$ in the invariant,
if all activities $f \in \adm(s)$ eventually leave the invariant,
there exists one such activity $f$ and a time $\delta \in \sp(f, \loc(s))$
%if $\pair{\loc(s), f(\delta)}$ is not in the invariant, then
%there exists $\delta' < \delta$ such that 
such that $s' = \pair{\loc(s), f(\delta)}$ 
is in the invariant and there is an uncontrollable transition 
$e \in \edgu$ such that $s' \xto{e} s''$.
%admits at least one (non-Zeno) run
%starting from any state that is reachable from an initial state.
If a hybrid automaton is non-Zeno and non-blocking, we say that
it is \emph{well-formed}.
In the following, all hybrid automata are assumed to be well-formed.

\begin{example}
Consider the LHAs in Figure~\ref{fig:wf}.
The fragment in Figure~\ref{fig:wf1} is well-formed, because
the system may choose derivative $\dot{x}=0$
and remain indefinitely in location $l$.
The fragment in Figure~\ref{fig:wf2} is also well-formed, because
the system cannot remain in $l$ forever, but an uncontrollable
transition leading outside is always enabled. 
Finally, the fragment in Figure~\ref{fig:wf3} is not well-formed, because
the system cannot remain in $l$ forever, and no uncontrollable
transition is enabled.
\end{example}

\begin{figure}[h]
\centering
\subfigure[Well-formed.]{
   \begin{tikzpicture}
     \node[ha, inner sep=0pt] (s0) [label={above:$l$}] 
     {$\begin{array}{c}x\in [0,1] \\ \dot{x}\in[-1,1]\end{array}$};
 \end{tikzpicture}
 \label{fig:wf1}
}
\subfigure[Well-formed.]{
   \begin{tikzpicture}[node distance=2.5cm, auto, shorten >=2pt, shorten <=2pt]
     \node[ha, inner sep=0pt] (s0) [label={above:$l$}] 
          {$\begin{array}{c}x\in [0,1] \\ \dot{x}\in[1,2]\end{array}$};
     \node[ha] (s1) [right of=s0] {...};

     \path[-stealth'] (s0) edge [dashed] node [above] {u} (s1);
 \end{tikzpicture}
 \label{fig:wf2}
}
\subfigure[Not well-formed.]{
   \begin{tikzpicture}[node distance=2.5cm, auto, shorten >=2pt, shorten <=2pt]
     \node[ha, inner sep=0pt] (s0) [label={above:$l$}] 
          {$\begin{array}{c}x\in [0,1] \\ \dot{x}\in[1,2]\end{array}$};
     \node[ha] (s1) [right of=s0] {...};
     \path[-stealth'] (s0) edge node [above] {c} (s1);
 \end{tikzpicture}
 \label{fig:wf3}
}
 \caption{Three LHA fragments. Locations contain the invariant (first line)
 and the flow constraint (second line). Solid (resp., dashed) edges represent
 controllable (resp., uncontrollable) transitions. Guards are $\mathtt{true}$.}
 \label{fig:wf}
\end{figure}
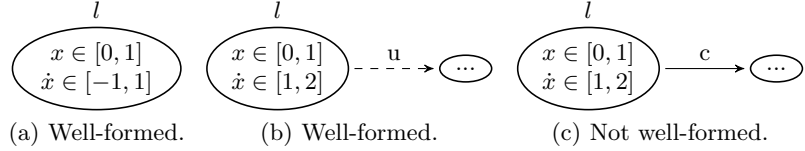

\paragraph{Strategies.}
%\mynote{Why not having strategies choose either an action or a delay?
%Check ``The element of surprise''.}
%
A \emph{strategy} is a function
$\sigma: S \to 2^{\edgc \cup \{\bot\}} \setminus \emptyset$,
where $\bot$ denotes the null action.
Notice that our strategies are \emph{non-deterministic} and \emph{memoryless}
(or \emph{positional}).
A strategy can only choose a transition which is allowed by the automaton.
Formally, for all $s\in S$, if $e \in \sigma(s) \cap \edgc$, 
then there exists $s'\in S$ such that $s\xto{e} s'$.
Moreover, when the strategy chooses the null action,
it should continue to do so for a positive amount of time,
along each activity that remains in the invariant.
If all activities immediately exit the invariant,
the above condition is vacuously satisfied.
%
%Formally, define the set $\urgent \subseteq \Inv$ of states
%$s$ such that there exists an activity $f \in \adm(s)$ 
%that immediately leaves the invariant, i.e.,
%there exists $\delta > 0$ such that for all 
%$0 < \delta' \leq \delta$ it holds $\pair{\loc(s),f(\delta)} \not\in \Inv$.
%
Formally, if $\bot \in \sigma(s)$, %either $s \in \urgent$ or
for all $f \in \adm(s)$ there exists $\delta>0$ such that
for all $0 < \delta' < \delta$ it holds 
$\delta' \not\in \sp(f, \loc(s))$ or  
$\bot \in \sigma(\pair{\loc(s),f(\delta')})$.
This ensures that the null action
is enabled in right-open regions, so that there is an earliest instant
in which a controllable transition becomes mandatory. 

Notice that a strategy can always choose the null action.
The well-formedness condition ensures that the system can always
evolve in some way, be it a timed step or an uncontrollable transition.
In particular, even if we are on the boundary of the invariant
we allow the controller to choose the null action, because,
in our interpretation, %the invariant is part of the system specification,
it is not the responsibility of the controller
to ensure that the invariant is not violated.
%The non-blocking assumption ensures that when \emph{all} activities
%immediately leave the invariant, an uncontrollable transition is enabled.

We say that a run like~\eqref{eq:run} is \emph{consistent} with a strategy 
$\sigma$ if for all $0\leq i< \length(r)$ the following conditions hold:
\begin{itemize}
\item for all $\delta \geq 0$ such that
$\sum_{j=0}^{i-1} \delta_j \leq \delta < \sum_{j=0}^{i} \delta_j$,
we have $\bot \in \sigma(r(\delta))$;
\item if $e_i \in \edgc$ then $e_i \in \sigma(s_i')$.
\end{itemize}
We denote by $\runs(s, \sigma)$ the set of runs starting from the state $s$
and consistent with the strategy $\sigma$.
The following result ensures that each strategy
has at least one run that is consistent with it,
otherwise the controller may surreptitiously satisfy the safety
objective by blocking the system.
The result can be proved by induction by considering that:
as long as the strategy chooses the null action,
the system may continue along one of the activities that remain
within the invariant; if a state is reached from which 
all activities immediately leave the invariant,
the well-formedness assumption ensures that 
there exists an uncontrollable transition that is enabled;
finally, if the strategy chooses a discrete transition, that 
transition is enabled.

\begin{theorem}
Given a well-formed hybrid automaton,
for all strategies $\sigma$ and initial states $s \in \Inits$,
there exists a run that starts from $s$ and is consistent with $\sigma$.
\end{theorem}

\paragraph{Safety control problem.}
Given a hybrid automaton $H$ and a set of states $T\subseteq\Invs$,
the \emph{safety control problem} asks whether
there exists a strategy $\sigma$ such that, for all initial states $s \in \Inits$,
all runs $r \in \runs(s, \sigma)$
it holds $\states(r) \subseteq T$. \\
We call the above $\sigma$ a \emph{winning strategy}.

%\paragraph{Reachability control problem.}
%Given a hybrid automaton $H$ and a set of target states $T \subseteq \Inv$, 
%the \emph{reachability control problem} asks whether
%there exists a strategy $\sigma$ such that, for all initial states $s \in \Init$
%and all runs $r \in \runs(s, \sigma)$,
%it holds $\states(r) \cap T \neq \emptyset$.

\section{Safety Control} \label{sec:safety}

In this section, we consider a fixed hybrid automaton and we present
a sound and complete procedure to solve the safety control problem.

\subsection{The Abstract Algorithm}

We start by defining some preliminary operators.  For a set of states
$A$ and $x \in \{ u, c\}$, let $\pretag_x^{\rm m}(A)$ (for \emph{may
  predecessors}) be the set of states where some discrete transition
belonging to $\edg_x$ is enabled, which leads to $A$, and let $\overline{A}$ be the set complement of $A$.  Analogously,
let $\pretag_x^{\rm M}(A) = \pretag_x^{\rm m}(A) \setminus
\pretag_x^{\rm m}(\overline{A})$ (the \emph{must predecessors}) be the
set of states where all enabled discrete transitions belonging to
$\edg_x$ lead to $A$, and there is at least one such transition
enabled.

\begin{theorem} \label{thm:main-loop}
The answer to the safety control problem for safe set $T\subseteq\Invs$
is positive if and only if
$$\Inits \subseteq \nu W \qdot T \cap \cpre(W),$$
where $\cpre$ is the controllable predecessor operator below.
\end{theorem}

\paragraph{Controllable predecessor operator.}
For a set of states $A$, the operator $\cpre(A)$ returns the set of
states from which the controller can ensure that the system remains in
$A$ during the next joint transition.
This happens if for all
activities chosen by the environment and all delays $\delta$, one of
two situations occurs:
\begin{itemize}
\item either the systems stays in $A$ up to time $\delta$, while all
  uncontrollable transitions enabled up to time $\delta$ (included)
  also lead to $A$, or
\item there exists a time $\delta' < \delta$, such that the
  system stays in $A$ up to time $\delta'$, all uncontrollable
  transitions enabled up to time $\delta'$ (included) also lead to
  $A$, and the controller can issue a transition at time $\delta'$
  leading to $A$.
\end{itemize}
To improve readability, for a set of states $A$, an activity $f$, and
a time delay $\delta\geq 0$ (including infinity), we denote by
$\while(A,f,\delta)$ the set of states from where following the
activity $f$ for $\delta$ time units keeps the system in $A$ all the
time, and any uncontrollable transition taken meanwhile also leads
into $A$.  Formally,
\begin{multline*}
\while(A, f, \delta) = \Bigl\{ s \in S \,\Bigl|\,
\forall 0\leq \delta'\leq\delta:
\pair{\loc(s), f(\delta')} \in A \setminus \pre{u}{m}(\overline{A}) \Bigr\}.
\end{multline*}
We can now formally define the $\cpre$ operator and prove Theorem~\ref{thm:main-loop}.
\begin{align*}
\cpre(A) &= 
\Bigl\{ s \in S \,\Bigl|\, \forall f \in \adm(s), \delta \in \sp(f, \loc(s)) :
s\in\while(A,f,\delta)\\ 
&\text{or } \exists 0\leq\delta' < \delta :
s \in \while(A,f,\delta') \text{ and }
\pair{\loc(s), f(\delta')} \in \pre{c}{m}(A) \Bigr\}.
\end{align*}

\begin{proof}{}
[{\bf \emph{if}}]
We shall first build a winning strategy in two steps. Let $W^* = \nu W
\qdot T \cap \cpre(W)$ and let $\sigma$ be a strategy defined as follows,
for all states $s$:
\begin{itemize}
\item $\bot \in \sigma(s)$ and
\item if $s \xto{e} s'$, $s,s'\in W^*$ and $e\in\edgc$, then $e \in
  \sigma(s)$.
\end{itemize}
% \begin{itemize}
% \item if $s \xto{e} s'$, $s'\in W^*$ and $e\in\edgc$, then $e \in
%   \sigma(s)$
% \item if either $s\in \urgent \cap \pre{u}{M}(W^*)$ or $\exists
%   \delta>0$, $\forall f\in\adm(s), 0\leq \delta'\leq \delta$,
%   $\pair{\loc(s),f(\delta')}\in W^*\cap T$, then $\bot\in\sigma(s)$.
% \end{itemize}
% and for all $s\in S\setminus W^*$, $\sigma(s)$ is any subset of
% $\edgc\cup\{\bot\}$ compatible with the definition of strategy (such
% a nonempty subset must exists as a consequence of well-formedness of
% the automaton).
%
While $\sigma$ is clearly a strategy, it is not necessarily a winning
strategy, as it may admit runs which delay controllable
actions either beyond the safety set $W^*$ or beyond their
availability. We can, however, recover a winning strategy by
restricting $\sigma$ in appropriate ways. For all states $s\in S$ and
activities $f \in\adm(s)$, let
\begin{align*}
D_{f,s} = \big\{ &\delta > 0 \mid \\
&\forall 0 \leq \delta' \leq \delta :
\pair{\loc(s),f(\delta')} \in W^* \text{ and }
\sigma(\pair{\loc(s),f(\delta')}) \cap \edgc \neq \emptyset \big\}.
\end{align*}
denote the set of positive time units for which the system can follow
activity $f$, starting from $s$, always remaining in $W^*$ with some
controllable transition enabled and available to the controller.

Starting from $\sigma$, we can define a new strategy $\sigma'$ which
coincides with $\sigma$ on all the states, except for the states $s\in
W^*$ with $\edgc\cap\sigma(s) \neq \emptyset$, where it satisfies
$\sigma'(s)\subseteq\sigma(s)$ and the following two conditions:
\begin{itemize}
\item[a)] If there is $f \in\adm(s)$ such that $D_{f,s} = \emptyset$, then
  $\bot\not\in\sigma'(s)$;
\item[b)]
  \begin{minipage}[t]{.95\textwidth}for all $f \in\adm(s)$, if
    $D_{f,s} \not= \emptyset$,
%    $\exists \delta>0.$ $\forall 0\leq \delta'\leq \delta.$
%    $\pair{\loc(s),f(\delta')}\in W^*$, $\bot\in
%    \sigma(\pair{\loc(s),f(\delta')})$ and $\edgc \cap
%    \sigma(\pair{\loc(s),f(\delta')}) \not= \emptyset$,\\
%
    then there exists a $\delta\in D_{f,s}$ with $\bot\not\in
    \sigma'(\pair{\loc(s),f(\delta)})$ and $\forall 0\leq \delta''<
    \delta$, $\bot\in
    \sigma'(\pair{\loc(s),f(\delta'')})$.\end{minipage}
\end{itemize}
Intuitively, the new strategy $\sigma'$ ensures that following any
activity from a state $s\in W^*$ in which some controllable action is
enabled, a controllable action will always be taken before none of
them is available and before leaving $W^*$.

To prove that $\sigma'$ is winning, we must show that for every $s\in
\Inits$ and every $r\in\runs(\sigma',s)$, $States(r)\subseteq T$. Let 
$$r = s_0 \xto{\delta_0,f_0} s_0' \xto{e_0}
      s_1 \xto{\delta_1,f_1} s_1' \xto{e_1} s_2 \ldots s_n \ldots
$$
be a run consistent with $\sigma'$.
The following properties can be proved:
\begin{enumerate}
\item \label{eq:timedtrans}
if $s_i \xto{\delta_i,f_i} s_i'$ occurs in $r$, with
$\delta_i>0$ and $s_i\in W^*$, then
for all $0\leq\delta'\leq\delta_i$,
it holds $\pair{\loc(s_i), f_i(\delta')}\in W^*$;
\item \label{eq:inftrans}
if $s_i \xto{\infty,f_i}$ occurs in $r$ and $s_i\in W^*$,
then for all $\delta'\geq 0$, it holds $\pair{\loc(s_i),f_i(\delta')}\in W^*$;
\item if $s_i \xto{e} s_i'$ occurs in $r$ and $s_i\in W^*$,
then $s_i'\in W^*$. \label{eq:immtrans}
\end{enumerate}

\begin{comment}
\begin{gather}
\begin{split}\label{eq:timedtrans}
\textit{if } s_i \xto{\delta_i,f_i} s_i' &  \textit{ occurs in $r$, with }
\delta_i>0 \textit{ and } s_i\in W^*, \textit{then} \\
&\textit{for all } 0\leq\delta'\leq\delta_i,
\pair{\loc(s_i), f_i(\delta')}\in W^*
%\textit{ and }s_i'\in W^*.
\end{split}\\
%
\begin{split}\label{eq:inftrans}
\textit{if } s_i \xto{\infty,f_i} & \textit{ occurs in $r$ and } s_i\in
W^*, \textit{ then} \\
&\textit{for all } \delta'\geq 0,  \pair{\loc(s_i),
  f_i(\delta')}\in W^*.
\end{split}\\
%
\textit{if } s_i \xto{e} s_i' \textit{ occurs in $r$ and } s_i\in W^*,
\textit{then } s_i'\in W^*. \label{eq:immtrans} \hspace{3cm}\  
\end{gather}
\end{comment}
%
We shall prove property (\ref{eq:timedtrans}), as (\ref{eq:inftrans}) can
be proved similarly. Since $\delta_i > 0$, by the consistency of $r$
with $\sigma'$, we have $\bot\in\sigma'(s_i)$. 
%Therefore, by definition of $\sigma'$, [[vale solo by well formedness
%and def of strategy, $\delta_i \in span(s,f_i)$]] either
%$s_i\in \urgent \cap \pre{u}{M}(W^*)$ or $\exists \delta>0$, $\forall
%f\in\adm(s_i), 0\leq \delta'\leq \delta$,
%$\pair{\loc(s_i),f_i(\delta')}\in W^*$.  Since $\delta_i>0$, we have
%$s_i=\pair{\loc(s_i),f_i(0)}\in W^*$.
Assume, by contradiction, that $\pair{\loc(s_i),f_i(\delta')}\not\in
W^*$ for some $0<\delta'<\delta_i$. Since $s_i\in
W^*=\cpre(W^*)$, then $s_i\in \while(W^*,f_i,\delta)$ for some
$\delta\in\reals^{\geq 0} \cup \{\infty\}$, and either $\delta =
\infty$ or $s_i \xrightarrow{\delta, f_i} s$ and $s\in\pre{c}{m}(W^*)$.

If $\delta\geq\delta'$, we have an immediate contradiction, since
it would imply $s_i\in\while(W^*,f_i,\delta')$ and, therefore,
$\pair{\loc(s_i),f_i(\delta')}\in W^*$. 
%Moreover, from
%$\bot\in\sigma'(\pair{\loc(s_i),f_i(\delta')})$
% and
%$\pair{\loc(s_i),f_i(\delta')}\not\in \urgent$, 
%we obtain $\pair{\loc(s_i),f_i(\delta')}\in W^*$.

Assume, then, $\delta < \delta'$.  
%Since
%$\pair{\loc(s_i),f_i(\delta)}=s \not\in \urgent$, 
Then $\pair{\loc(s_i), f_i(\delta)} \in \pre{c}{m}(W^*)$, i.e.,
  $\pair{\loc(s_i),f_i(\delta)} \xto{e} s'$ for some $e \in \edgc$
  and $s'\in W^*$.  Therefore, both $e\in \sigma'(\pair{\loc(s_i),
  f_i(\delta)})$ and, by the consistency of $r$ with $\sigma'$,
  $\bot\in\sigma'(\pair{\loc(s_i), f_i(\delta)})$.
Since $\bot\in\sigma'(\pair{\loc(s_i), f_i(\delta)})$, by definition
of $\sigma'$ the premise of property~{\it a}) cannot hold. Therefore,
by property {\it b}), there must be a $\delta \leq \delta^* <
\delta'$ with $\bot\not \in
\sigma'(\pair{\loc(s_i),f_i(\delta^*)})$. On the other hand,
the consistency of $r$ requires that $\bot \in
\sigma'(\pair{\loc(s_i),f_i(\hat{\delta})})$ for all $0\leq
\hat{\delta}<\delta_i$, which is a contradiction. Therefore, for all
$0\leq\delta'<\delta_i$, $\pair{\loc(s_i), f_i(\delta')}\in W^*$.

%In addition, consistency of $r$ also ensures that for all $0\leq
%\delta'<\delta_i$, $\bot\in\pair{\loc(s_i), f_i(\delta')}$ and, by
%definition of $\sigma'$, $\pair{\loc(s_i), f_i(\delta')}\in W^*\cap
%T$. 
Finally, to prove that $s_i'\in W^*$ we can proceed again by
contradiction, assuming $s_i'\not\in W^*$. Let $0 < \delta'<\delta_i$,
then $\pair{\loc(s_i), f_i(\delta')}\in W^*$. Therefore,
$\pair{\loc(s_i), f_i(\delta')}\in \cpre(W^*)$ and there exists
$\delta'\leq\delta^*<\delta_i$ with $\pair{\loc(s_i),
f_i(\delta')}\in \while(W^*,f_i,\delta^*)$ and $\pair{\loc(s_i),
f_i(\delta^*)}\in \pre{c}{m}(W^*)$. Hence, there is a controllable
transition $e\in\edgc$ enabled in $\pair{\loc(s_i), f_i(\delta^*)}$
and leading to $W^*$. As a consequence,
$\{e,\bot\}\subseteq\sigma(\pair{\loc(s_i), f_i(\delta^*)})$ and, by
condition {\it b}), $\bot\not\in \sigma'(\pair{\loc(s_i),
f_i(\overline{\delta})})$, for some $\delta^* < \overline{\delta}
< \delta_i$, which contradicts consistency of $r$ with $\sigma'$,
hence $s_i'\in W^*$.

Let us consider property (\ref{eq:immtrans}). We have two cases. If
$e\in\edgc$, then the consistency of $r$ ensures that $e\in\sigma'(s_i)$
which, by definition of $\sigma'$, requires that $s_{i+1}\in
W^*$. Assume then that $e\in\edgu$. Then $\bot\in\sigma'(s_i)$. 
%We
%have to cases depending on whether $s_i\in \urgent \cap
%\pre{u}{M}(W^*)$ or not.
%
%Since $s_i\in \urgent \cap \pre{u}{M}(W^*)$, by definition of
%$\pre{u}{M}(W^*)$, any uncontrollable transition enabled in $s_i$ must
%lead to $W^*$, hence the thesis.
%
%In the other case, 
Since $s_i\in W^* = \cpre(W^*)$, it must hold $s_i \in
\while(W^*,f,0)$, for every $f\in\adm(s_i)$. This, in turn, ensures
that $s_i\in W^*\setminus \pre{u}{m}(\overline{W^*})$, therefore, all
the uncontrollable transitions enabled in $s_i$ lead to $W^*$. Hence
the thesis.

%An immediate consequence of properties (\ref{eq:timedtrans}) and
%(\ref{eq:immtrans}) is that, since the starting state of $r$ is in
%$W^*$, for any $i\geq 0$, $s_i,s_i'\in W^*$.

To complete the proof, %showing that $States(r)\subseteq T$, 
notice that $W^*\subseteq T$ and $s_0\in Inits \subseteq W^*$. An easy
induction on the length of $r$, using properties (\ref{eq:timedtrans}),
(\ref{eq:inftrans}) and (\ref{eq:immtrans}), gives the result.

%For any $\delta\geq 0$ there exists a unique index $i$ in $r$ such
%that $\sum_{j=0}^{i-1} \delta_j \leq \delta < \sum_{j=0}^{i}
%\delta_j$. This implies both $\delta_i> 0$ and $0\leq
%\delta-\sum_{j=0}^{i-1} \delta_j < \delta_i$, and, by definition,
%$r(\delta) = \pair{\loc(s_i),f_i(\delta-\sum_{j=0}^{i-1})}$. Moreover,
%we have $s_i\in W^*$ and, by property (\ref{eq:timedtrans}) and
%(\ref{eq:inftrans}), we can immediately conclude $r(\delta)\in T$.

[{\bf \emph{only if}}]
Let $s \not\in W^*$, we prove that for all
strategies there is a run that starts in $s$, is consistent with the
strategy and leaves $T$. Let 

\begin{itemize}
\item $W_0 = T$, 
\item $W_{\alpha} = T \cap\cpre(W_{\alpha-1})$, for a successor ordinal $\alpha$, and 
\item $W_{\alpha}=\bigcap_{\alpha'<\alpha}W_{\alpha'}$ for a limit ordinal $\alpha$.
\end{itemize}

We proceed by induction on the smallest ordinal $\lambda$ such that $s\not\in W_\lambda$.
If $\lambda=0$, it holds $s \not\in T$ and the thesis is immediate.
%% We again proceed by induction, based on the maximum number
%% $k(s)$ of consecutive discrete transitions that can be taken from $s$.
%% Since the system is non-Zeno, $k(s)$ exists and is finite.
%% Fix a strategy $\sigma$ and let $r \in \runs(\sigma, s)$.
%% If $k(s)=0$, the run $r$ starts with a non-zero timed transition,
%% we have $r(0) = s \not\in T$, and we obtain the thesis.
%% Otherwise, $r$ starts with a discrete transition $s \xto{e} s'$.
%% Let $I = \ireach(T)$.
%% If $s' \not\in I$, it holds $k(s') < k(s)$, and by inductive hypothesis we obtain the thesis.
%% If instead $s' \in I$, we distinguish two cases:
%% First, assume that $e \in \edgc$.
%% Since $s \not\in I$, it also holds 
%% $s \not\in \pre{c}{m}(I) \setminus \pre{u}{m}(\overline{I})$, i.e.,
%% $s \in \overline{\pre{c}{m}(I)} \cup \pre{u}{m}(\overline{I})$.
%% We have $s \in \pre{u}{m}(\overline{I})$.
%% Hence, there is another run that is consistent with $\sigma$ and
%% starts with a discrete transition $s \xto{e'} s''$,
%% where $e \in \edgu$ and $s'' \in \overline{I}$.
%% We have $k(s'') < k(s)$ and by inductive hypothesis we obtain the thesis.

%% Next, assume that $e \in \edgu$.
%% Notice that $s \not\in \urgent \cap \pre{u}{M}(I)$.
%% If $s \not\in \urgent$, the environment is not forced to issue the transition $e$.
%% Hence, if $\bot \in \sigma(s)$, there is another run that is consistent with $\sigma$
%% and starts with a non-zero timed transition.
%% If $\bot \not\in \sigma(s)$, there is another run that starts with a controllable 
%% transition, a case that was already treated above.

We will show that if $\lambda>0$ then $\lambda$ cannot be a limit ordinal. Assume by contradiction that $\lambda$ is a limit ordinal. 
Since $\lambda$ is the smallest ordinal such that $s\not\in W_\lambda$, we have $s\in W_\alpha$, for all $\alpha<\lambda$: this means that $s\in\bigcap_{\alpha<\lambda} W_\alpha$. But, since $\lambda$ is a limit ordinal, $W_\lambda=\bigcap_{\alpha<\lambda}W_\alpha$ and we have that $s\in W_\lambda$, obtaining a contradiction.

Otherwise, if $\lambda>0$ is a successor ordinal, we have $s \in W_{\lambda-1} \setminus W_\lambda$ and $s \not\in
\cpre(W_{\lambda-1})$.  According to the definition of $\cpre$, there exists
an activity $f \in \adm(s)$ and $\delta\in span(s,f)$ such that $s
\not\in \while(W_{\lambda-1}, f, \delta)$ and for all $0\leq \delta' <
\delta$ either $s \not\in \while(W_{\lambda-1}, f, \delta')$ or
$\pair{\loc(s), f(\delta')} \not\in \pre{c}{m}(W_{\lambda-1})$.

%Since $s \not\in \while(W_{n-1}, f, \infty)$, there exists $\delta \geq 0$
%such that $s \not\in \while(W_{n-1}, f, \delta)$.
Let $\delta^*$ be the infimum of those $\delta'$ such that
$s\not\in\while(W_{\lambda-1}, f, \delta')$, i.e.,
\begin{equation} \label{eq:deltastar}
\delta^* = \inf \{ \delta \mid s \not\in \while(W_{\lambda-1}, f, \delta) \}.
\end{equation}
Clearly $0\leq \delta^*\leq \delta$ and, for all $0\leq\overline{\delta} <
\delta^*$, $\pair{\loc(s), f(\overline{\delta})} \not\in
\pre{c}{m}(W_{\lambda-1})$.
%property~\eqref{eq:onlyif} holds.  
Hence, any controllable transition enabled in $\pair{\loc(s),
  f(\overline{\delta})}$, for any such $\overline{\delta}$, leads
outside $W_{\lambda-1}$.
%, and the environment is not forced to take any
%uncontrollable transition leading into $W_{\lambda-1}$.
Therefore, any strategy choosing a controllable transition in some 
of the states $\pair{\loc(s), f(\overline{\delta})}$ 
has a consistent run leading outside $W_{\lambda-1}$.
By inductive hypothesis, we obtain the thesis.

If, on the other hand, the strategy allows the controller to stay
inactive in all those states, there is a consistent run that reaches
$\delta^*$. Then we have two cases.  If $\delta^*$ is in fact the
minimum of the above set, according to the definition of $\while$,
there exists $\delta_1 < \delta^*$ such that $\pair{\loc(s),
 f(\delta_1)} \in \overline{W_{\lambda-1}} \cup \pre{u}{m}(\overline{W_{\lambda-1}})$.
Therefore, since the controller may
not act before $\delta^*$ along this strategy, there is a consistent
run that reaches $\pair{\loc(s), f(\delta_1)}$, which either is in
$\overline{W_{\lambda-1}}$ or reaches it after an uncontrollable
transition. In both cases, the thesis follows from the inductive hypothesis.

Finally, we have the case in which $\delta^*$ is the infimum but not
the minimum of the above set. In this case $0\leq \delta^* < \delta$
and $\pair{\loc(s), f(\overline{\delta})} \not\in
\pre{c}{m}(W_{\lambda-1})$, for all $0 \leq \overline{\delta} \leq
\delta^*$.  Consider the choice of $\sigma$ in state $\pair{\loc(s),
  f(\delta^*)}$.  If $\bot \not\in \sigma(\pair{\loc(s),
  f(\delta^*)})$, the controller issues a discrete move which leads
into $\overline{W_{\lambda-1}}$.  If, instead, $\bot \in
\sigma(\pair{\loc(s), f(\delta^*)})$, since $\delta^* < \delta \in
span(s,f)$, by the definition of strategy $\sigma$ will keep choosing
$\bot$ for a non-zero amount of time $\gamma$.
By~\eqref{eq:deltastar}, there exists $\delta^* < \hat\delta <
\delta^* + \gamma$ such that $s \not\in \while(W_{\lambda-1}, f,
\hat\delta)$.  As a consequence, there is a consistent run that
reaches a state which either is in $\overline{W_{\lambda-1}}$ or reaches it
after an uncontrollable transition.  Once again, the thesis is
obtained by inductive hypothesis.  \qed
\end{proof}

\subsection{Computing the Predecessor Operator on LHAs}\label{sec:compute}

In this section, we show how to compute the value of the predecessor
operator on a given set of states $A$, assuming that the hybrid automaton
is a LHA and that we can compute
the following operations on arbitrary polyhedra $G$ and $G'$:
%\begin{itemize}
the Boolean operations $G\cup G$, $G \cap G$, and $\overline{G}$;
the topological closure $\clos{G}$ of $G$; finally,
for a given location $l\in\Loc$, the \emph{pre-flow} of $G$ in $l$:
$$\epref{G} = \{ u \in \Val(X) \mid \exists \delta\geq 0, c \in \Flow(l) : 
u + \delta \cdot c  \in G \}.$$
Notice that, for two convex polyhedra $P$ and $P'$, if $P\subseteq P'$
then $\epref{P}\subseteq\epref{P'}$ (monotonicity), and
$\epref{(\epref{P})}\,=\epref{P}$ (idempotence).

In the following, we proceed from the basic components of $\cpre$
to the full operator. 
Given a set of states $A$ and a location $l$, we
denote by $\proj{A}{l}$ the projection of $A$ on $l$, i.e. 
$\{v\in \Val(X) \mid \pair{l,v}\in A\}$.
For all $A \subseteq \Invs$ and $x \in\{c, u\}$, it holds:
\begin{align*}
\pre{x}{m}(A) = \Invs \cap \bigcup_{(l,\mu,l') \in \edg_x} \mu^{-1}(\proj{A}{l'}),
\end{align*}
where $\mu^{-1}(Z)$ is the pre-image of $Z$ w.r.t. $\mu$.
%polyhedron obtained from $\mu$ by switching the
%dimensions representing the variables in $X$ with those representing
%the variables in $X'$.
%The set $\urgent$ can be computed by applying the following equation.
%$$
%\urgent = \bigcup_{l \in \Loc, P \in \conv{\proj{\Inv}{l}},
% P' \in \conv{\overline{\proj{\Inv}{l}}}}
%P \cap \clos{P'} \cap \epref{P'}.
%$$
%
We also introduce the
auxiliary operator $\rwamay$ (\emph{may reach while avoiding}). Given
a location $l$ and two sets of variable valuations $U$ and $V$,
$\rwamay_l(U,V)$ contains the set of valuations from which the
continuous evolution of the system \emph{may} reach $U$ while avoiding
$V \cap \overline{U}$~\footnote{In $\atl$ notation, we have $\rwamay_l(U,V) \equiv
  \team{env} (\overline{V} \cup U) \U U$, where $env$ is the player
  representing the environment.}.
Notice that on a dense time domain this is not equivalent to reaching $U$ 
while avoiding $V$: If an activity avoids $V$ in a right-closed interval,
and then enters $U \cap V$, the first property holds, while the latter does not.
%When the system reaches $U$, it is allowed to be also in $V$.
%This is connected to the assumption that if both 
%
Formally, we have:
\begin{multline*}
\rwamay_l(U, V) = \Bigl\{ u \in \Val(X) \,\Bigl|\,
\exists f \in \adm(\pair{l,u}),\, \delta\geq 0: \\
f(\delta)\in U \text{ and }
\forall\, 0\leq \delta'<\delta: f(\delta')\in \overline{V}\cup U \Bigr\}.
\end{multline*}
An algorithm for effectively computing $\rwamay$ is presented in the
next section, while the following lemma
states the relationship between $\cpre$ and $\rwamay$.
Intuitively, consider the set $B_l$ of valuations $u$ such that from state
$\pair{l, u}$ the environment can take a discrete transition leading
outside $A$, and the set $C_l$ of valuations $u$ such that from
$\pair{l, u}$ the controller can take a discrete transition into $A$.
We use the $\rwamay$ operator to compute the set of valuations
from which there exists an activity that either leaves $A$
or enters $B_l$, while staying in the invariant and avoiding $C_l$.
These valuations do not belong to $\cpre(A)$, as the environment
can violate the safety goal within (at most) one discrete transition. \\
We say that a set of states $A \subseteq S$ is \emph{polyhedral}
if for all $l \in \Loc$, the projection $\proj{A}{l}$ is a polyhedron.
\begin{lemma} \label{lem:cpreS}
For all polyhedral sets of states $A \subseteq \Invs$, we have
\begin{equation} \label{eq:cpreS}
\cpre(A) = \bigcup_{l \in \Loc} \{ l \} \times \Big( \proj{A}{l} \setminus
\rwamay_l\big( \proj{\Invs}{l} \cap \big( \overline{\proj{A}{l}} \cup B_l \big),
               C_l \cup \overline{\proj{\Invs}{l}} \big) \Big),
\end{equation}
where
$B_l = \proj{\pre{u}{m}\big(\overline{A}\big)}{l}$ and
$C_l = \proj{\pre{c}{m}(A)}{l}$.
%$D_l = \rwamust_l\Big(\overline{\Inv}, \overline{\pre{u}{M}(A)}\Big)$.
\end{lemma}
\begin{proof}{}
In the following, let $I_l = \proj{\Invs}{l}$.

[$\subseteq$] Let $s = \pair{l, u} \in \cpre(A)$ and let $f \in \adm(s)$.
By definition, $0 \in \sp(f, l)$ and hence $s \in \while(A,f,0)$.
In particular, this implies that $s \in A$ and $u \in \proj{A}{l}$.

Assume by contradiction that $s$ does not belong to the r.h.s. of~\eqref{eq:cpreS}.
Since $u \in \proj{A}{l}$, it must be 
$$u \in \rwamay_l\big( I_l \cap \big( \overline{\proj{A}{l}} \cup B_l \big),
               C_l \cup \overline{I_l} \big).$$
Then, by definition there exists $f^* \in \adm(s)$ and $\delta^* \geq 0$
such that:
\emph{(i)} $f^*(\delta^*) \in I_l \cap (\overline{\proj{A}{l}} \cup B_l)$, and
\emph{(ii)} for all $0 \leq \delta < \delta^*$ it holds
$f^*(\delta) \in I_l \cap \big( \overline{C_l} \cup \overline{\proj{A}{l}} \cup B_l \big)$. In particular, this implies that $\delta^*$ belongs to $\sp(f^*,l)$.
On the other hand, if we apply the definition of $\cpre(A)$ to the activity $f^*$,
we obtain that for all $\delta \in \sp(f^*, l)$
either $s \in \while(A,f^*,\delta)$ or
there exists $\delta' < \delta$ such that 
$s \in \while(A,f^*,\delta')$ and $\pair{l, f^*(\delta')} \in \pre{c}{m}(A)$.
This implies that either $f^*(\delta^*) \in \proj{A}{l} \cap \overline{B_l}$
or there exists $\delta' < \delta^*$ such that
$f^*(\delta') \in \proj{A}{l} \cap \overline{B_l} \cap C_l$, which is a contradiction.

[$\supseteq$] Let $l \in \Loc$ and $u \in \proj{A}{l} \setminus
\rwamay_l\big( I_l \cap (\overline{\proj{A}{l}} \cup B_l), 
               C_l \cup \overline{I_l} \big)$.
By complementing the definition of $\rwamay$, we obtain that
for all activities $f$ that start from $s = \pair{l,u}$ and for all times
$\delta\geq 0$, either 
$f(\delta) \in \overline{I_l} \cup (\proj{A}{l} \cap \overline{B_l})$
or there exists $\delta'<\delta$ such that 
$$
f(\delta') \in 
\Big( \overline{I_l} \cup \big( \proj{A}{l} \cap \overline{B_l}\big)\Big)
\cap (C_l \cup \overline{I_l}) 
= \overline{I_l} \cup \big(\proj{A}{l} \cap \overline{B_l} \cap C_l \big)
\eqdef E_l.
$$
First, assume that for all $\delta\geq 0$ it holds 
$f(\delta) \in \overline{I_l} \cup (\proj{A}{l} \cap \overline{B_l})$.
In this case, for all $\delta\in\sp(f,l)$,
the point $f(\delta)$ belongs to $\proj{A}{l} \cap \overline{B_l}$.
In other words, $s \in \while(A, f, \delta)$
and hence $s \in \cpre(A)$.

Otherwise, there exists $\delta'$ 
such that $f(\delta') \in E_l$. Let $\delta^*$ be the
infimum of the $\delta'$ with the above property,
i.e., $\delta^* = \inf \{ \delta' \mid f(\delta') \in E_l \}$.
Notice that it holds 
$f(\delta) \in \overline{I_l} \cup \big(\proj{A}{l} \cap \overline{B_l}\big)$
for all $\delta \leq \delta^*$, which implies $s \in \while(A,f,\delta^*)$.
If there exists $\delta \leq \delta^*$ such that $f(\delta) \in \overline{I_l}$,
again we conclude that for all $\delta\in\sp(f,l)$
it holds $f(\delta) \in \proj{A}{l} \cap \overline{B_l}$
and hence $s \in \cpre(A)$.
In the rest of the proof, we can assume that
$f(\delta) \in I_l$ for all $\delta \leq \delta^*$,
and therefore $\delta^* \in \sp(f, l)$.

If $\delta^*$ is in fact the minimum of the above set, i.e.,
$f(\delta^*) \in E_l$, then according to the current assumptions
we have in particular $f(\delta^*) \in C_l = \proj{\pre{c}{m}(A)}{l}$.
Accordingly, $s \in \cpre(A)$.
Finally, we are left with the case in which $f(\delta^*) \not\in E_l$.
By definition, in any neighbourhood of $\delta^*$ there is a time
$\delta$ such that $f(\delta) \in E_l$.
Due to the fact that $E_l$ is a polyhedron
and that $f$ is differentiable, there exists $\delta' > \delta^*$
such that $f(\delta) \in E_l$ for all 
$\delta^*< \delta \leq \delta'$.
Therefore, $s \in \while(A,f,\delta')$,
and $\pair{l, f(\delta')} \in C_l =
\pre{c}{m}(A)$.
Again, we obtain that $s \in \cpre(A)$.
\qed
\end{proof}

\subsection{Computing the $\rwamay$ operator on LHAs}

In this section, we consider a fixed location $l$.
Given two polyhedra $G$ and $G'$, we define their \emph{boundary}
%\mynote{Sostituire border con entry region.}
to be
$$ \bound(G, G') = (\clos{G} \cap G') \cup (G \cap \clos{G'}).$$
We can compute $\rwamay$ by the following fixpoint characterization.
\begin{theorem} \label{thm:may} For all locations $l$ and sets of
  valuations $U$, $V$, and $W$, let
\begin{equation} \label{eq:our-fixp}
\tau(U,V,W) = U \cup
\bigcup_{P\in\conv{\overline{V}}} \bigcup_{P'\in \conv{W}}
\Bigl( P \cap \epref{\bigl( \bound(P, P') \cap \epref{P'} \bigr)} \Bigr).
\end{equation}
We have $\rwamay_l(U, V) = \mu W \qdot \tau(U,V,W)$.
\end{theorem}
Roughly speaking, $\tau(U,V,W)$ represents the set of points which either
belong to $U$ or do not belong to $V$ and can reach $W$ along a
straight line which does not cross $V$.  We can interpret the fixpoint
expression $\mu W \qdot \tau(U,V,W)$ as an incremental refinement of an
under-approximation to the desired result.  The process starts with
the initial approximation $W_0 = U$.  One can easily verify that $U
\subseteq \rwamay_l(U,V)$.  Additionally, notice that $\rwamay_l(U,V)
\subseteq U \cup \overline{V}$.  The equation refines the
under-approximation by identifying its \emph{entry regions}, i.e., the
boundaries between the area which \emph{may} belong to the result
(i.e., $\overline{V}$), and the area which already belongs to it
(i.e., $W$).  
That is, let $P\in\conv{\overline{V}}$ and $P' \in\conv{W}$,
let $b=\bound(P,P')$,
we call $R = b \cap \epref{P'}$ an entry region \emph{from $P$ to $P'$},
and also an entry region \emph{of $W$}.
%
%Let now $E$ and $G$ be any two polyhedra\mynote{Polyhedra or sets of
%  convex polyhedra?} and $b=\bound(P, P')$ be the {\em adjacency
%  region} between a polyhedron $P \in \conv{E}$ and a polyhedron
%$P'\in\conv{G}$. 
The set $R$ contains the points of $b$ that may reach $P'$ 
by following the continuous evolution of the system. 
Hence, the system may move from $P$ to $P'$ through $R$.
Moreover, the set $R' = P \cap \epref{R}$ contains the points of $P$
that can move to $P'$ through $R$.  
%Accordingly, we say that $b'$ is an \emph{entry region} of $G$ and we
%use the notation $Int(R)$ (resp. $Ext(R)$) to indicate the internal a
%polyhedron $P'$ (resp. the external polyhedron $P$).
%
%Additionally, a proper entry region must be traversed by
%the system continuous evolution, in the direction of $W$.  To check
%whether a boundary may be traversed by the system flow, we use the $\epref{}$ operator.
%Notice that $\epref{P'}$ contains 
%also points $u$ that need to exit the invariant in order to reach $P'$.
%This is intentional, as an exit border may be such because it is adjacent
%to the complement of the invariant.
%
Any point in $\overline{V}$ that may reach an entry region (without
reaching $V$ first) must be added to the under-approximation, since it
belongs to $\rwamay_l(U, V)$.

\paragraph{Proof of Theorem~\ref{thm:may}.}
First, we show that the $\tau$ operator is monotonic w.r.t.\ its third
argument, so that the least fixpoint $\mu W \qdot \tau(U,V,W)$ is well
defined.
\begin{lemma}
For all polyhedra $U$, $V$, and $W\subseteq W'$, it holds $\tau(U,V,W)
\subseteq \tau(U,V,W')$.
\end{lemma}
\begin{proof}
Assume for simplicity that $\conv{W} \subseteq \conv{W'}$.  Then, it
is sufficient to observe that, for all $P \in \conv{\overline{V}}$,
the expression $\bigcup_{P'\in \conv{W}} \bigl( P \cap \epref{\bigl(
  \bound(P, P') \cap \epref{P'} \bigr)} \bigr)$ is monotonic
w.r.t.\ $W$, since it is composed by monotonic operators.  \qed
\end{proof}
The following lemma allows us to switch from arbitrary activities
to piecewise straight lines, within Lemma~\ref{lem:complete}.
\begin{lemma}[\cite{WongToi}] \label{lem:straight} For all locations
  $l$, and valuations $u$ and $v$, if there is an activity $f \in
  \adm(\pair{l,u})$ and a time $\delta\geq 0$ such that $f(\delta) =
  v$ avoiding $V$, then there is a finite sequence of straightline
  activities leading from $u$ to $v$, each avoiding $V$.
\end{lemma}
Theorem~\ref{thm:may} is an immediate consequence of the following
two lemmas.
\begin{lemma} %[Completeness of~\eqref{eq:our-fixp}]
  \label{lem:complete}
  For all locations $l$ and polyhedra $U$ and $V$, it holds
  $\rwamay_l(U, V) \subseteq \mu W \qdot \tau(U,V,W)$.
\end{lemma}
\begin{proof}
  Let $u \in \rwamay_l(U, V)$ and $W^* = \mu W \qdot \tau(U,V,W)$.  By
  definition, $u \in \overline{V} \cup U$. If $u$ belongs to $U$, then
  it belongs to $W^*$ by definition. If $u$ belongs to
  $\overline{V}\setminus U$, there must be an activity that starts in
  $u$ and reaches a point $u' \in U$ without visiting $V\setminus U$.
  By Lemma~\ref{lem:straight}, there is a finite sequence of
  straightline segments leading from $u$ to $u'$ and avoiding
  $V\setminus U$.  Let $u_0, u_1, \ldots, u_k$ be the corresponding
  sequence of intermediate corner points, where $u_0 = u$ and $u_k =
  u'$.
%Clearly, this path is entirely contained in $\overline{V}$.
  We proceed by induction on $k$.  If $k=0$, it holds $u = u' \in U$,
  and the thesis is trivially true.  If $k>0$, we apply the inductive
  hypothesis to $u_1$, and we obtain that $u_1 \in W^*$.  Consider the
  straight path from $u_0 \in \overline{V}\setminus U$ to $u_1 \in
  W^*$.  This path crosses into $W^*$ in a given point $v$.  Formally,
  $v$ is the first point along the path which belongs to $\clos{W^*}$.
  Hence, there is at least one convex polyhedron $P'\in\conv{W^*}$
  such that $v\in\clos{P'}$. If there is more than one such
  polyhedron, pick the one that contains at least one point of the
  straight path from $v$ to $u_1$.  In this way, we have $v \in
  \epref{P'}$.

  Let $n$ be the number of convex polyhedra in
  $\conv{\overline{V}\cup U}$ that are crossed by the straight path
    from $u_0$ to $v$.  We start a new induction on $n$.  If $n=1$,
    the whole line segment from $u_0$ to $v$ is contained in a given
    $P \in \conv{\overline{V}\cup U}$. Hence, $v \in \bound(P, P')$,
    where $P'$ is a suitable element of $\conv{W^*}$.  Summarizing, we
    have $v \in \bound(P, P') \cap \epref{P'}$ and $u_0 \in \epref{\{
      v \}}$.  We conclude that $u_0 \in W^*$.  If $n>1$, we split the
    straight path from $u_0$ to $v$ into $n$ segments, defined by the
    intermediate points $v_1,\ldots,v_{n-1}$, and we apply the
    inductive hypothesis to $v_1$, obtaining that $v_1 \in W^*$.
    Finally, we use an argument analogous to the one for $n=1$ to
    conclude that $u_0\in W^*$.  \qed
\end{proof}

\begin{lemma} %[Soundness of~\eqref{eq:our-fixp}]
\label{lem:sound}
For all locations $l$ and polyhedra $U$ and $V$, it holds
$\rwamay_l(U, V) \supseteq \mu W \qdot \tau(U,V,W)$.
\end{lemma}
\begin{proof}
  It suffices to show that $\rwamay_l(U, V)$ is a fixpoint of $r$,
  i.e., $\rwamay_l(U, V) = \tau(U,V,\rwamay_l(U, V))$.  Let $u \in
  \tau(U,V,\rwamay_l(U, V))$, we shall prove that $u \in \rwamay_l(U,
  V)$.  If $u \in U$, the thesis is obvious.  Otherwise, there exist
  $P \in \conv{\overline{V}}$ and $P'\in\conv{\rwamay_l(U, V))}$ such
  that $u \in P \cap \epref{\bigl( \bound(P, P') \cap \epref{P'}
    \bigr)}$.  Hence, there is a straightline activity $f \in
  \adm(\pair{l, u})$ that reaches a point $v \in \bound(P, P') \cap
  \epref{P'}$, while staying in $P \subseteq \overline{V}$.  If $v \in
  P'$, we are done, as we have found an activity from $u$ to
  $\rwamay_l(U,V)$ which avoids $V\setminus U$.  Otherwise, $v$
  belongs to $\clos{P'} \cap \epref{P'}$ and, therefore, can reach
  some point $x\in P'$ through an arbitrarily small flow step along
  some activity. Since $P'\subseteq \rwamay_l(U, V))$, any other
  possible point $z$ between $v$ and $x$ along the activity belongs to
  $P'$ and, therefore, cannot belong to $V\setminus U$. Hence,
  $v\in \rwamay_l(U, V)$ and, consequently, the so does $u$.

  Finally, let $u \in \rwamay_l(U, V)$, we show that $u \in
  \tau(U,V,\rwamay_l(U, V))$. First, notice that $u \in U \cup
  \overline{V}$.  If $u\in U$, the thesis is obvious. Otherwise, there
  exist $P\in\conv{\overline{V}}$ and $P'\in\conv{\rwamay_l(U, V)}$
  such that $u \in P \cap P'$. Therefore, we also have $u \in
  \bound(P,P')$ and $u \in \epref{P'}$. By~\eqref{eq:our-fixp}, we
  obtain the thesis.  \qed
\end{proof}

\paragraph{Termination.}
The following theorem states the termination of the fixpoint procedure
defined in Theorem~\ref{thm:may}.
\begin{theorem}\label{thm:term}
The fixpoint procedure for $\rwamay$ defined in Theorem~\ref{thm:may}
terminates in a finite number of steps.
\end{theorem}

\begin{comment}
  \begin{proof} We shall give only a sketch of the proof. The complete
    proof is reported in the appendix.  Notice that
    $\conv{\overline{V}}$ and $\conv{U}$ are finite sets of convex
    polyhedra, therefore so is the number of initial entry regions of
    $\conv{U}$. Intuitively, the fixpoint procedure of
    Theorem~\ref{thm:may} applies the refinement steps according to a
    breadth-first policy, starting from the initial entry regions. In
    other words, at every iteration each entry region discovered so
    far is employed in a refinement step. It can be proved (see the
    appendix) that the number of different entry regions is finite and
    that every possible entry region is discovered after a number of
    iterations bounded by the size of
    $\conv{\overline{V}}$. Therefore, after at most
    $\vert\conv{\overline{V}}\vert +1$ iterations of the procedure,
    the fixpoint is reached.
\end{proof}
\end{comment}
%
In order to prove Theorem~\ref{thm:term}, we shall need some
additional definitions and notation.
%We say that two convex polyhedra $P$ and $P'$ are \emph{adjacent} if
%$\bound(P, P') \neq \emptyset$~\footnote{Notice that any two
%  overlapping polyhedra $P$ and $P'$ are considered adjacent to one another.}.
% Let $b$ be the boundary between two adjacent convex
%polyhedra $P \in \conv{E}$ and $P'\in\conv{G}$.  The set $R = b \cap
%\epref{P'}$ represents the points of $b$ that may reach $P'$ by
%following the continuous evolution of the system.  Hence, the system
%may move from $P$ to $P'$ through $R$.  
%\begin{lemma}
%Let $R = \bound(P, P') \cap \epref{P'}$. 
%For each $u \in R$ there is an activity $f\in\adm(\pair{l,u})$
%such that for all $\epsilon>0$ there exists $\delta<\epsilon$ 
%such that $f(\delta) \in P'$.
%\end{lemma}
Given two polyhedra $E$ and $G$ and two convex polyhedra $P
\in\conv{E}$ and $P'\in\conv{G}$, if the entry region $R$ from $P$ to
$P'$ is not empty, we write $G \rstep{P, R} G'$, where $\conv{G'} =
\conv{G} \cup \{P \cap \epref{R}\}$, to denote a {\it refinement
  step}.
%\footnote{To Fix: $R$ viene aggiunto solo se non \`e gi\`a
%  contenuto in qualche poliedro di $\conv{G}$.}
%
The following lemma can easily be proved exploiting idempotence
and monotonicity of $\epref{}$.
\begin{lemma} \label{lem:step2} Assume $G \rstep{P, R} G'$. For all
  entry regions $R'$ of $G'$ that are not entry regions of $G$ it
  holds $R' \subseteq \epref{R}$.
\end{lemma}
\begin{proof}
By definition of entry region,
$R'=\bound{(P',P\cap\epref{R})}\cap\epref{(P\cap\epref{R})}$,
with $P'\in E$ and $P\cap\epref{R}\in G$.
Hence, we can write $R'\subseteq\epref{(P\cap\epref{R})}$. Moreover, from
($P\cap\epref{R})\subseteq\epref{R}$ and by monotonicity and
idempotence properties of $\epref{}$ it follows that
$\epref{(P\cap\epref{R})}\subseteq\epref{R}$. Hence the thesis
$R'\subseteq\epref{(P\cap\epref{R})}\subseteq\epref{R}$. \qed
\end{proof}
Intuitively, the fixpoint procedure to compute $\rwamay$ applies, at
each iteration $k\geq 1$, all the refinement steps of the form $G
\rstep{P, R} G'$, with $E = \overline{V}$ and $G =
\riter{k-1}{U}{V}{U}$ (where $\riter{0}{U}{V}{U} = U$ and
$\riter{i+1}{U}{V}{U} = \rfix{U}{V}{\riter{i}{U}{V}{U}}$) for every
entry region $R$ of the current under-approximation $G$, following a
breadth-first policy. The following lemma make the relationship
between (sequences of) refinement steps and the $\tau(\cdot)$ operator
precise.
\begin{lemma}\label{lem:fixref} If  $\pi = G_0 \rstep{P_1, R_1} G_1 \rstep{P_2, R_2} \ldots
  \rstep{P_m, R_m} G_m$ is a sequence of refinement steps with $R$
  entry region of $G_l$, then $R$ is an entry region of
  $\riter{m}{G_0}{\overline{E}}{G_0}$.
\end{lemma}
\begin{proof} Let $\overline{\pi} = G_0 \rstep{P_1, R_1} G_1
  \rstep{P_2, R_2} \ldots \rstep{P_k, R_k} G_k$ be the shortest prefix
  of $\pi$ such that $R$ is entry region of $G_k$. Clearly, $k\leq
  m$. We now proceed by induction on $k$. If $k=0$, then $R$ is entry
  region of $G_0 = \riter{0}{G_0}{\overline{E}}{G_0}$ and, by
  monotonicity of the operator $\tau$, the thesis holds.\\
  Assume $k > 0$. Since $R$ is entry region of $G_k$ but not in
  $G_{k-1}$ and $\conv{G_k} = \conv{G_{k-1}}\cup \{P_k \cap
  \epref{R_{k}}\}$, it must be $R = \bound(P,(P_k \cap \epref{R_{k}}))
  \cap \epref{(P_k \cap \epref{R_{k}})}$, with $P\in \conv{E}$ and
  $R_{k}$ entry region in $G_{k-1}$.
% and $(P_k \cap \epref{R_{k}}) \in \conv{G_k} \setminus
%  \conv{G_{k-1}}$. 
  By induction hypothesis, $R_k$ is entry region of
  $\riter{k-1}{G_0}{\overline{E}}{G_0}$. Since 
%$R_k$ is an entry
%  region of $\riter{k-1}{G_0}{\overline{E}}{G_0}$ and 
  $P_k\in \conv{E}$, by definition of $\tau$ we have $(P_k \cap
  \epref{R_{k}}) \in
  \rfix{G_0}{\overline{E}}{\riter{k-1}{G_0}{\overline{E}}{G_0}}
  =\riter{k}{G_0}{\overline{E}}{G_0}$. Therefore, $R$ is an entry
  region in $\riter{k}{G_0}{\overline{E}}{G_0}$. Again, by
  monotonicity of $\tau$, the thesis follows. \qed
%  By induction hypothesis, $R_{k}$ is entry region of
%  $\tau(G_0,\overline{E},G_0)^{k-1}$. Moreover, ...
\end{proof}
We shall now show that the number of different entry regions employed
by the fixpoint procedure for $\rwamay$ is finite and that the number
of its iterations is bounded, thus establishing termination of the
procedure itself.\\
We need first some properties of sequences of refinement steps.  Given
a sequence $\pi = G_0 \rstep{P_1, R_1} G_1 \rstep{P_2, R_2} \ldots
\rstep{P_k, R_k} G_k$, $\last(\pi)$ denotes $G_k$.  Moreover, given a
convex polyhedron $R$, let $\prune(\pi, R)$ be the sequence obtained
from $\pi$ by removing all edges which depend on $R$, i.e. such that
$R_i\subseteq \epref{R}$.
Formally, $\prune(\pi, R) = G_0 \rstep{P'_1, R'_1} G_1' \rstep{P'_2,
  R'_2} \ldots \rstep{P'_m,R'_m} G'_m$ is the largest subsequence of
$\pi$ such that $R'_i\neq R$, for all $1\leq i\leq m$.  Clearly, we
have $m \leq k$.
The following lemma states that $\prune(\pi, R)$ preserves all the
entry regions of $\last(\pi)$ that do not depend on $R$.
%The following lemma states that if there is a sequence of refinement
%steps $\pi$ starting from $G_0$, where $R'$ is an entry region of
%$\last(\pi)$, and $R'$ does not depend on $R$ (i.e., $R' \not\subseteq
%\epref{R}$), the subsequence $\pi'$ obtained by pruning from $\pi$ all
%edges whose label depends on $R$ still has $R'$ as entry region of in
%$\last(\pi')$

\begin{lemma} \label{lem:prun} Let $\pi = G_0 \rstep{P_1,
    R_1} G_1 \rstep{P_2, R_2} \ldots \rstep{P_k,
    R_k} G_k$ be a sequence of refinement steps and let $R$ be an
  entry region of $G_k$, such that $R \not\subseteq \epref{R_1}$.
  Then, there exists a subsequence $\pi'$ of $\prune(\pi, R_1)$ such
  that $R$ is an entry region of $\last(\pi')$.
\end{lemma}
\begin{proof}
We proceed by induction on $k$.  If $k=1$, we have that $R$ is an
entry region of $G_1$, with $R\subseteq \epref{R_1}$. By
Lemma~\ref{lem:step2} $R$ must be entry region of $G_0$.

If $k>1$, let $j$ be the smallest index such that $R$ is an entry
region in $G_j$. If $j=0$, we are done. Otherwise, by
Lemma~\ref{lem:step2} we have $R \subseteq \epref{R_j}$. Consequently,
$R_j \not\subseteq \epref{R_1}$ (otherwise, by monotonicity it would
hold $R\subseteq \epref{R_1}$).  Apply the inductive hypothesis to the
prefix $G_0 \rstep{P_1,R_1} G_1 \rstep{P_2,R_2}
\ldots \rstep{P_{j-1},R_{j-1}} G_{j-1}$ and to $R_j$. We
obtain that there exists a sequence $\pi'$ that starts from $G_0$,
does not use $R_1$, and ends in a polyhedron $G'$ such that $R_j$ is
an entry region of $G'$.  Hence, for the sequence $\pi'
\rstep{P_j,R_j} G''$ we have that $R$ is an entry region of
$G''$ and we obtain the thesis. \qed
\end{proof}

%\begin{lemma}
%For all polyhedra $G$ and borders $b$, if
%there is a derivation from $G$ to $b$ then there is a derivation
%from $G$ to $b$ such that no edge is labeled with a border belonging
%to the same face as $b$.
%\end{lemma}
%\begin{proof}
%\end{proof}

% \begin{lemma}\label{lem:star}
% Let $G$ be a polyhedron and let $R$ $R'$ and $P$ be convex polyhedra
% with $R'\subseteq\epref{R}$ and $(P\cap\epref{R})\in\conv{G}$. If
% $G\rstep{P,R'} G'$ then $G=G'$.
% \end{lemma}
% \begin{proof}
% By definition, $G\rstep{P,R'}G'$ implies
% $\conv{G'}=\conv{G}\cup(P\cap\epref{R'})$.  Using the monotonicity and
% idempotence properties of $\epref{}$ on the hypothesis
% ($R'\subseteq\epref{R}$), we can write $\epref{R'}\subseteq\epref{R}$
% and, therefore, $(P\cap\epref{R'})\subseteq(P\cap\epref{R})$.  By the
% hypothesis $(P\cap\epref{R})$ is already in $\conv{G}$, therefore, by
% definition of $G\rstep{P,R'}G'$, we conclude $G=G'$.
% \qed
% \end{proof}

We are now ready to state the main property relating entry regions and
sequences of refinement steps.

\begin{lemma}\label{lemma:bound}
  Let $\pi = G_0 \rstep{P_1, R_1} G_1 \rstep{P_2,
    R_2} \ldots \rstep{P_n, R_n} G_n$ be a sequence of
  refinement steps and let $R$ be an entry region of $G_n$.  Then,
  there exists a subsequence $\pi'=G_0 \rstep{P'_1, R'_1}
  G'_1 \rstep{P'_2, R'_2} \ldots \rstep{P'_m, R'_m}
  G'_m$, such that:
%\begin{enumerate}
$R$ is an entry region of $G'_m$ and
$P'_i \neq P'_j$ for all $1\leq i < j\leq m$.
%\end{enumerate}
\end{lemma}
\begin{proof}
Let $\overline{\pi} = G_0 \rstep{P_1, R_1} G_1
\rstep{P_2, R_2} \ldots \rstep{P_k, R_k} G_k$ be the
shortest prefix of $\pi$ such that $R$ is an entry region of $G_k$. We
proceed by induction on $k$. If $k=0$ or $k=1$, the thesis immediately
follows.  If $k > 1$,
%let
%$\pi=G_0\rstep{P_1,R_1}G_1\rstep{P_2,R_2} \ldots
%\rstep{P_{k-1},R_{k-1}}G_{k-1}\rstep{P_k, R_k}G_k$
then $R_k$ is an entry region of $G_{k-1}$ and $R$ is an entry region
in $G_k$.  Since $k$ is the first index for which $R$ is an entry
region in $G_k$, we also have $R\subseteq P_k\cap\epref{R_k}$.  We can
now apply the inductive hypothesis on
$G_0\rstep{P_1,R_1}G_1\rstep{P_2,R_2} \ldots
\rstep{P_{k-1},R_{k-1}}G_{k-1}$ to obtain the subsequence
$\pi'=G_0\rstep{P'_1,R'_1}G'_1 \rstep{P'_2,R'_2}
\ldots \rstep{P'_h,R'_h}G'_h$, where $P'_i\neq
P'_j$, for all
%$i\neq j$ with %
$1\leq i < j\leq h$, and $R_k$ is still an entry region of
$G'_h$.\\ 
Hence, $\pi^* = \pi'\rstep{P_k,R_k}G'_{h+1}$ is a sequence of
refinement steps, and $R\subseteq P_k\cap\epref{R_k}$ implies that $R$
is an entry region of $G'_{h+1}$.  Assume $P'_j=P_k$ for some $1\leq
j\leq h$.  Considering the subsequence $\hat{\pi} =
G'_{j-1}\rstep{P'_j,R'_j}G'_{j}\rstep{P'_{j+1},
  R'_{j+1}}\ldots\rstep{P_h,R_h}G'_h$, two cases may occur:

\begin{enumerate}
\item if $R_k\not\subseteq \epref{R'_j}$, then substituting
  $\prune(\hat{\pi},R_j')$ for $\hat{\pi}$ in $\pi^*$ we obtain, by
  Lemma~\ref{lem:prun}, the desired sequence of refinement steps;
\item if $R_k\subseteq\epref{R'_j}$, then the subsequence
  $G_0\rstep{*}G'_{j}$ of $\pi'$ is the desired sequence. Indeed, by
  idempotence of $\epref{}$, $R_k \subseteq \epref{R'_j}$ implies
  $\epref{R_k}\subseteq\epref{R'_j}$. Since $P'_j=P_k$, we also have
  that $P_k \cap \epref{R_k}\subseteq P'_j \cap
  \epref{R'_j}$. Therefore, $R\subseteq P_k\cap\epref{R_k} \subseteq
  P'_j \cap \epref{R'_j}$. Hence, $R$ is an entry region of
  $G'_j$. \qed
% because, by Lemma~\ref{lem:star} we have that
%  $G'_s\rstep{P_k,R_k}G'_{s+1}$ implies $G'_{s+1}=G'_s$ and
%  then $R$ is an entry region of $G'_s$.
\end{enumerate}
\end{proof}

\noindent An immediate consequence of the previous lemma is that for
any entry region $R$ there is a sequence $\pi$ of refinement steps
discovering $R$ (i.e. with $R$ entry region of $\last(\pi)$) whose
length is bounded by $\vert \conv{E}\vert$. \\
We can now establish termination of the fixpont procedure to compute
$\rwamay$.

\paragraph{Proof of Theorem \ref{thm:term}.}
Notice that $\conv{\overline{V}}$ and $\conv{U}$ are finite sets of
convex polyhedra, therefore so is the number of initial entry regions
of $\conv{U}$. The fixpoint procedure of Theorem~\ref{thm:may} applies
the refinement steps in a breadth-first manner starting from the
initial entry regions. Therefore, in every iteration each entry region
discovered so far is employed in a refinement step. As a consequence
of Lemma~\ref{lemma:bound}, taking $E = \overline{V}$ and $G_0=U$, for
every entry region there is a sequence of refinement steps discovering
it, whose length is bounded by $\vert \conv{\overline{V}}
\vert$. Therefore, by Lemma~\ref{lem:fixref}, after at most $\vert
\conv{\overline{V}} \vert$ iterations of the procedure all the entry
regions have been discovered, and the fixpoint is reached at the next
iteration. \qed

\subsection{Previous Algorithms} \label{sec:previous}

In the literature, the standard reference for safety control of linear
hybrid systems is~\cite{WongToi}. 
The model and the abstract algorithm are essentially similar to ours,
except that, differently from our semantics, the states from which
a discrete transition is taken are subject to the safety constraint.
%While the model is essentially the
%same as ours, their abstract algorithm, though similar, is not
%complete, as the fixpoint procedure they propose neither accounts for
%the effect of zero time transitions nor of the invariants.
As to the computation of $\cpre$, they introduce an operator
$\mathit{flow\_avoid}$, which
corresponds to our $\rwamay$ operator. They propose to compute
$\rwamay_l(U,V)$ using the following fixpoint formula:
\begin{equation} \label{eq:wongtoi}
\bigcup_{U' \in \conv{U}} \bigcap_{V' \in \conv{V}} 
\Bigl( \mu W \qdot U' \cup \bigcup_{P\in \conv{\overline{V'}}} 
\bigl( \clos{P} \cap \overline{V'} \cap \epref{( W \cap P)} \bigr) \Bigr) 
\end{equation}
A simple example, however, shows that~\eqref{eq:wongtoi} is different
from (in particular, larger than) $\rwamay_l(U,V)$ when $V$ is
non-convex.  Consider the example in Figure~\ref{fig:wongtoi}, where
$U$ is the gray box on top and $V$ is the union of the two white
boxes.  Formula~\eqref{eq:wongtoi} treats the two convex parts of $V$
separately.  As a consequence, the result is the area covered by
stripes.  However, the correct results should not include the area
within the thick border (in red-colored stripes), because any point in
that region cannot prevent hitting one of the two convex parts of $V$.

\begin{figure}
\centering
\subfigure[Wong-Toi]{\label{fig:wongtoi}
\includegraphics[width=1.7in]{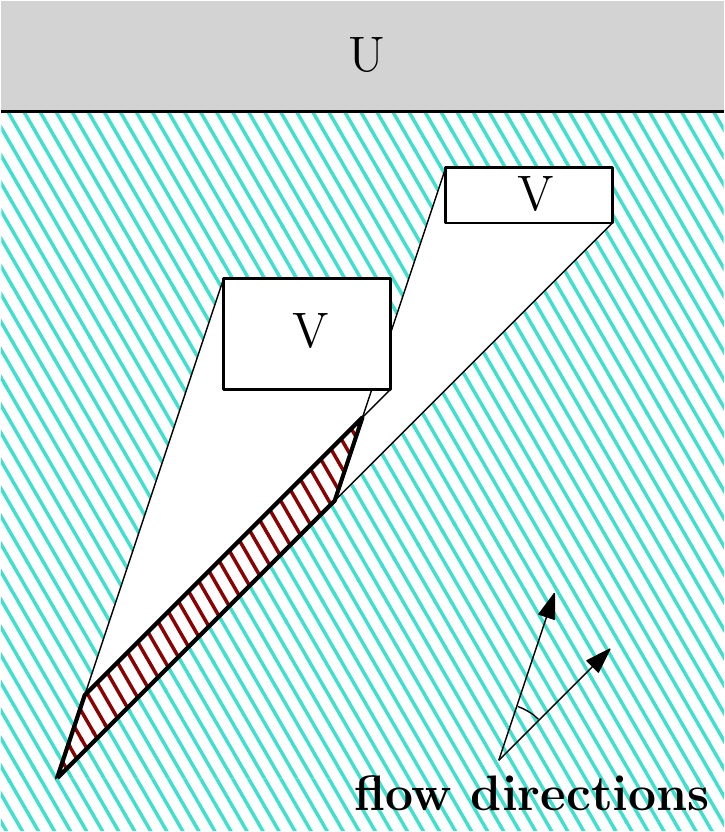}
}
\subfigure[HoneyTech]{\label{fig:honeytech}
\includegraphics[width=1.7in]{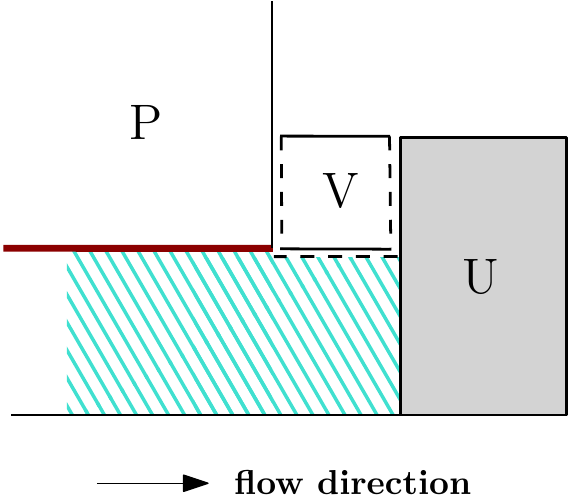}
}
\caption{Mistakes in previous fixpoint characterizations.}
\label{fig:mistakes}
\end{figure}

In~\cite{honeytech}, Deshpande et al. report about an implementation of
Wong-Toi's algorithm in the tool {\sc HoneyTech}, obtained as an extension
of HyTech.  The fixpoint formula that is meant to capture
$\rwamay_l(U,V)$ is the following:
\begin{equation} \label{eq:honey}
\mu W \qdot U \cup \bigcup_{P\in \conv{\overline{V}}}
\Bigl( P \cap \epref{\bigl( \clos{W} \cap \clos{P} \cap \overline{V} \cap \epref{W} \bigr)} \Bigr)
\end{equation}
Compared to~\eqref{eq:wongtoi}, formula~\eqref{eq:honey} correctly
treats the case of non-convex $V$. However, it suffers from another
issue, pertaining the distinction between topologically open and
closed polyhedra. Consider the example in Figure~\ref{fig:honeytech},
where $U$ is the gray box, $V$ is the white box, and dashed lines
represent topologically open sides of polyhedra.
The result of applying formula~\eqref{eq:honey} is the area covered
by near-vertical stripes. This area includes the thick solid line that
starts from a corner of $V$.
Indeed, if $W$ is the union of $U$ and the striped region and $P\in \conv{\overline{V}}$,
the thick line is exactly $\clos{W} \cap \clos{P} \cap \overline{V} \cap \epref{W}$.
However, this line does not belong to $\rwamay_l(U,V)$,
because all its points cannot avoid hitting $V$ before eventually reaching $U$.

\section{Experiments with PHAVer+} \label{sec:experiments}
In this section we show experiments about safety control.
We implemented the procedure showed in the previous section on the top of the open-source 
tool PHAVer~\cite{phaver05}.
%For the purpose of evaluating the present paper, a binary pre-release of our implementation,
%that we call PHAVer+, can be downloaded at {\tt http://people.na.infn.it/mfaella/cav11}.
The experiments were performed on an Intel Xeon (2.80GHz) PC.

\paragraph{Truck Navigation Control.}
The following example, the Truck Navigation Control (\emph{TNC}), is derived from the work~\cite{honeytech}.
Consider an autonomous toy truck, which is responsible for avoiding some 2 by 1 rectangular pits. 
The truck can take 90-degrees left or right turns: the possible directions are North-East (NE), 
North-West (NW), South-East (SE) and South-West (SW). One time unit must pass between two changes of  direction. The control goal consists in avoiding the pits.
% Notice that the \textit{TNC} proposed in~\cite{honeytech} is limited to one turn only, while our analysis is extended to the complete case (an unlimited number of turns is allowed).
Figure~\ref{fig:TNC} shows the hybrid automaton that models the system: there is one location for each direction, where the derivative of the position variables ($x$ and $y$) are set according to the corresponding direction. The variable $t$ represents a clock ($\dot{t}=1$) that enforces a one-time-unit wait between turns.

Figure~\ref{fig:iter} shows the three iterations needed to compute the fixpoint in Theorem~\ref{thm:main-loop}, in the case of two pits.
%: the fixpoint procedure performs three iterations (the third one reaches the fixpoint).
The safe set is the white area, while the gray region 
contains the points wherefrom it is not possible to avoid the pits.

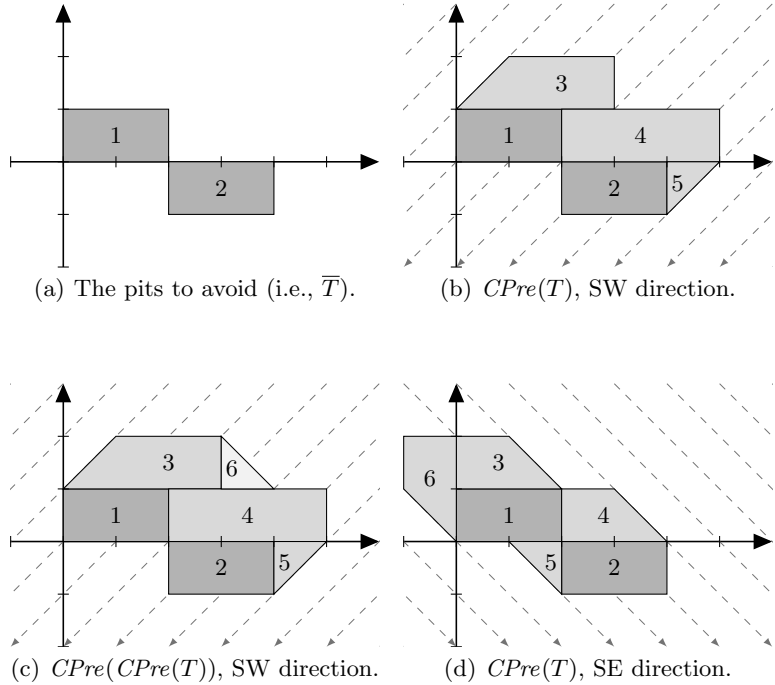
\begin{figure}[h]
\centering
\subfigure[The pits to avoid (i.e., $\overline{T}$).]{
\label{fig:ini}
\begin{tikzpicture}[cprefig]
% area da visualizzare
\clip (-1,-2) rectangle (6,4);
% assi
\draw[-triangle 45,semithick] (0,-2) -- (0,3);
\draw[-triangle 45,semithick] (-1,0) -- (6,0);
% i poligoni
\filldraw[my_rect1] (0,0) rectangle (2,1);
\filldraw[my_rect1] (2,0) rectangle (4,-1);
\draw (1,0.5)   node {\small 1};
\draw (3,-0.5)  node {\small 2};
% tacche sugli assi
\foreach \x in {-1,0,...,5} \draw (\x,-3pt) -- (\x,3pt);
\foreach \y in {-2,-1,...,2} \draw (-3pt,\y) -- (3pt,\y);
\end{tikzpicture}
}
%
%%%%%%%%%%%%%%%%%%%%%%%%%% 2 %%%%%%%%%%%%%%%%%%%%%%%%%%%%%%
%
\subfigure[$\cpre(T)$, SW direction.]{\label{fig:first}
\begin{tikzpicture}[cprefig]
% area da visualizzare
\clip (-1,-2) rectangle (6,4);
% assi
\draw[-triangle 45,semithick] (0,-2) -- (0,3);
\draw[-triangle 45,semithick] (-1,0) -- (6,0);
% frecce diagonali
\foreach \x in {-1,0,...,11}
   \draw[dashed,-latex,very thin, black!55] (\x,3) -- +(-5,-5);
% i poligoni
\filldraw[my_rect1] (0,0) rectangle (2,1);
\filldraw[my_rect1] (2,0) rectangle (4,-1);
\filldraw[my_rect2] (0,1) -- (1,2) -- (3,2) -- (3,1) -- cycle;
\filldraw[my_rect2] (2,0) rectangle (5,1);
\filldraw[my_rect2] (4,0) -- (5,0) -- (4,-1) -- cycle;
\draw (1,0.5)   node {\small 1};
\draw (3,-0.5)  node {\small 2};
\draw (2,1.5)   node {\small 3};
\draw (3.5,0.5) node {\small 4};
\draw (4.2,-0.4) node {\small 5};
% tacche sugli assi
\foreach \x in {-1,0,...,5} \draw (\x,-3pt) -- (\x,3pt);
\foreach \y in {-2,-1,...,2} \draw (-3pt,\y) -- (3pt,\y);
\end{tikzpicture}
}
%
%%%%%%%%%%%%%%%%%%%%%%%%% 3 %%%%%%%%%%%%%%%%%%%%%%%%%%%%%%
%
\subfigure[$\cpre(\cpre(T))$, SW direction.]{
\label{fig:second}
\begin{tikzpicture}[cprefig]
% area da visualizzare
\clip (-1,-2) rectangle (6,4);
% assi
\draw[-triangle 45,semithick] (0,-2) -- (0,3);
\draw[-triangle 45,semithick] (-1,0) -- (6,0);
% frecce diagonali
\foreach \x in {-1,0,...,11}
   \draw[dashed,-latex,very thin, black!55] (\x,3) -- +(-5,-5);
% i poligoni
\filldraw[my_rect1] (0,0) rectangle (2,1);
\filldraw[my_rect1] (2,0) rectangle (4,-1);
\filldraw[my_rect2] (0,1) -- (1,2) -- (3,2) -- (3,1) -- cycle;
\filldraw[my_rect2] (2,0) rectangle (5,1);
\filldraw[my_rect2] (4,0) -- (5,0) -- (4,-1) -- cycle;
\filldraw[my_rect3] (3,1) -- (3,2) -- (4,1) -- cycle;
\draw (1,0.5)   node {\small 1};
\draw (3,-0.5)  node {\small 2};
\draw (2,1.5)   node {\small 3};
\draw (3.5,0.5) node {\small 4};
\draw (4.2,-0.4) node {\small 5};
\draw (3.2,1.4) node {\small 6};
% tacche sugli assi
\foreach \x in {-1,0,...,5} \draw (\x,-3pt) -- (\x,3pt);
\foreach \y in {-2,-1,...,2} \draw (-3pt,\y) -- (3pt,\y);
\end{tikzpicture}
}
%
%%%%%%%%%%%%%%%%%%%%%%%%%%% 4 %%%%%%%%%%%%%%%%%%%%%%%%%%%%%
%
\subfigure[$\cpre(T)$, SE direction.]{
\label{fig:firstSE}
\begin{tikzpicture}[cprefig]
% area da visualizzare
\clip (-1,-2) rectangle (6,4);
% assi
\draw[-triangle 45,semithick] (0,-2) -- (0,3);
\draw[-triangle 45,semithick] (-1,0) -- (6,0);
% frecce diagonali
\foreach \x in {-6,-5,...,6}
    \draw[dashed,-latex,very thin, black!55] (\x,3) -- +(+5,-5);
% i poligoni
\filldraw[my_rect1] (0,0) rectangle (2,1); % 1
\filldraw[my_rect1] (2,0) rectangle (4,-1); % 2
\filldraw[my_rect2] (0,1) -- (0,2) -- (1,2) -- (2,1) -- cycle; % 3
\filldraw[my_rect2] (2,0) -- (2,1) -- (3,1) -- (4,0) -- cycle; % 4
\filldraw[my_rect2] (1,0) -- (2,0) -- (2,-1) -- cycle; % 5
\filldraw[my_rect2] (-1,2) -- (-1,1) -- (0,0) -- (0,2) -- cycle; % 6
\draw (1,0.5)   node {\small 1};
\draw (3,-0.5)  node {\small 2};
\draw (.8,1.5)   node {\small 3};
\draw (2.8,0.5) node {\small 4};
\draw (1.8,-0.4) node {\small 5};
\draw (-.5,1.2) node {\small 6};
% tacche sugli assi
\foreach \x in {-1,0,...,5} \draw (\x,-3pt) -- (\x,3pt);
\foreach \y in {-2,-1,...,2} \draw (-3pt,\y) -- (3pt,\y);
\end{tikzpicture}
}
%
%%%%%%%%%%%%%%%%%%%%%%%%%%%%%%%%%%%%%%%%%%%%%%%%%%%%%%%%
%
\caption{Evolution of the fixpoint in the case of two pits.
All figures are cross-sections for $t=0$. Dashed arrows represent
flow direction.}
\label{fig:iter}
\end{figure}

The input safe region $T$ is the area outside 
the gray boxes 1 and 2 in Figure~\ref{fig:ini}.
The first iteration (Figure~\ref{fig:first}) computes $\cpre(T)$ and extends the unsafe set to those points (areas 3, 4, and 5) that will inevitably 
flow into the pits, before the system reaches $t=1$ and the truck can turn.
The second iteration (Figure~\ref{fig:second}) computes $\cpre(\cpre(T))$
and extends the unsafe set by adding the area 6:
those points may turn before reaching the pits, but after the turn they
end up in $\cpre(T)$ anyway (for instance, if turning left, 
they end up in area 4 of Figure~\ref{fig:firstSE}).
The third iteration reaches the fixpoint.

We tested our implementation on progressively larger versions of the truck model, by increasing the number of pits.
We also considered a version of TNC with non-deterministic continuous flow,
allowing some uncertainty on the exact direction taken by the vehicle.
Using an exponential scale, Figure~\ref{fig:perf} compares
the performance of our tool (solid line for deterministic model,
dashed line for non-deterministic) to the performance reported in~\cite{honeytech} (dotted line).
We were not able to replicate the experiments in~\cite{honeytech},
since {\sc HoneyTech} is not publicly available.

\begin{figure}[h!]
\centering
\subfigure[Hybrid Automaton for TNC.]{\label{fig:TNC}
\includegraphics[width=2.5in]{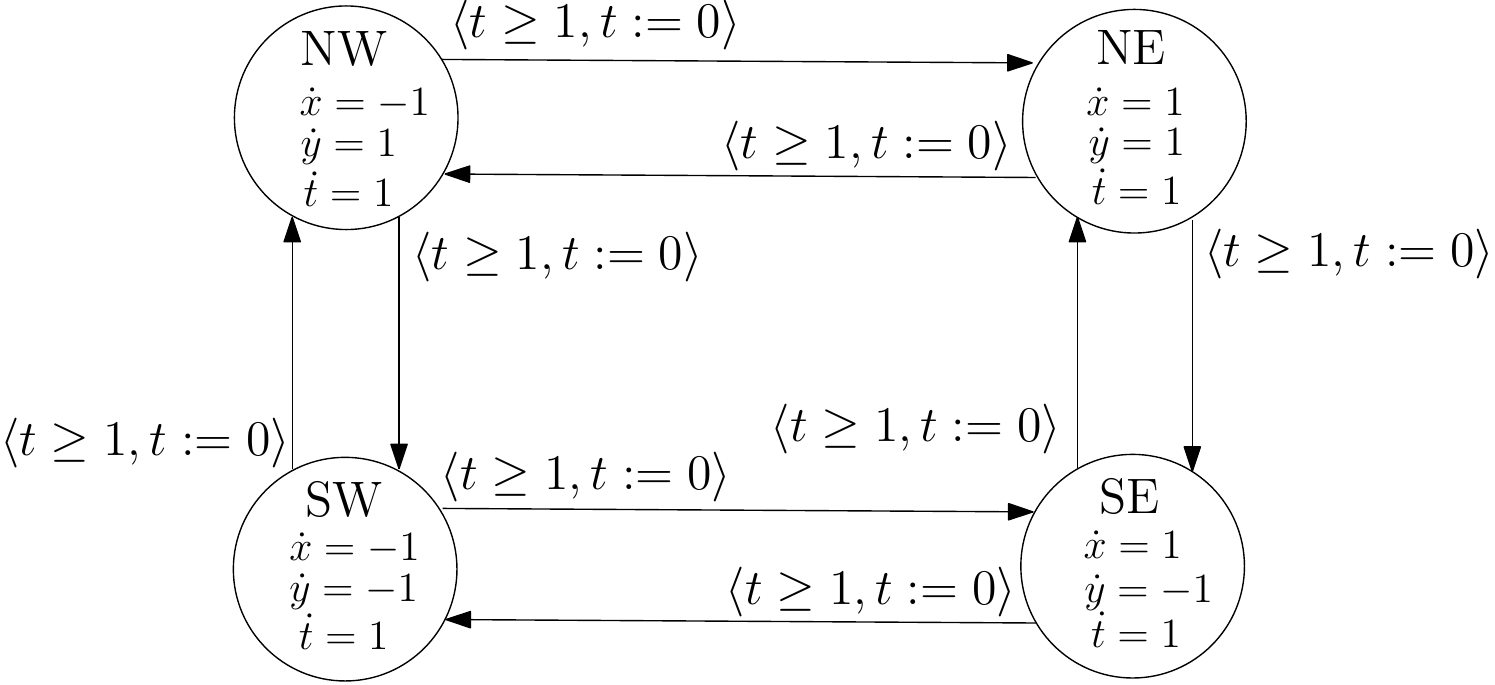}
}
\subfigure[Computation time as a function of the number of pits.]{\label{fig:perf}
\includegraphics[width=2.0in]{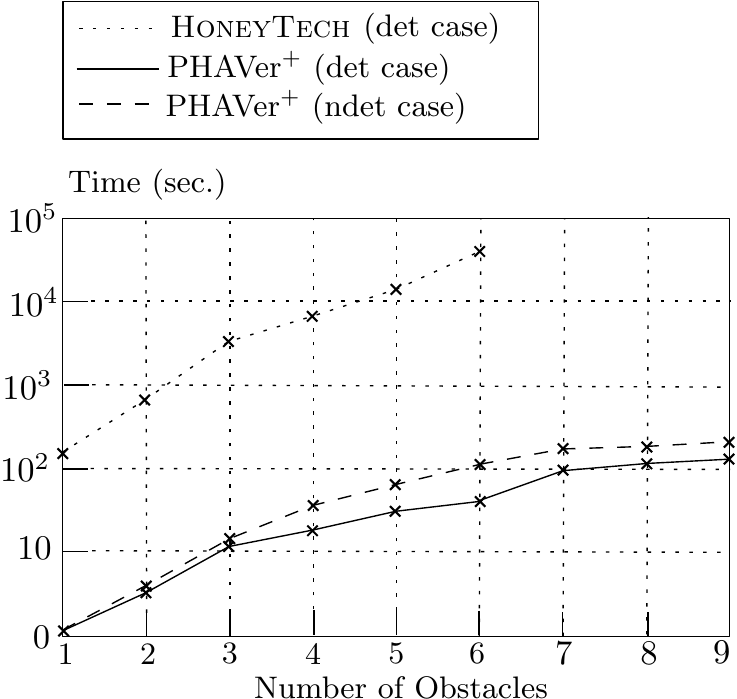}
}
\caption{Hybrid Automaton and performance for TNC.}
\label{fig:autandperf}
\end{figure}

Because of the different hardware used, only a qualitative comparison can be made:
going from 1 to 6 pits (as the case study in~\cite{honeytech}),
the run time of {\sc HoneyTech} shows an exponential behavior,
while our tool exhibits an approximately linear growth, as shown in
Figure~\ref{fig:perf}, where the performance of PHAVer+ is plotted 
up to 9 pits.

%\begin{table}[ht]
%\caption{Performance with respect to number of pits}\vspace{-.3cm}
%\center
%\begin{tabular}{| c | c | c | c | c | c | c | c | c | c | c |}
%\hline
%Pits (num.) & 1 & 2 & 3 & 4 & 5 & 6 & 7 & 8 & 9 & 10\\
%\hline
%Time (sec.) & 0.7 & 5.1 & 11.2 & 31.8 & 48.6 & 58.0 & 97.9 & 147.2 & 160.4 & 205.8 \\
%\hline
%\end{tabular}
%\label{tab:perf}
%\end{table}

\paragraph{Water Tank Control.}
Consider a system where two tanks --- A and B --- are linked by a one-directional valve \textit{mid} (from A to B). There are two additional valves: the valve \textit{in} to fill A and the valve \textit{out} to drain B. The two tanks are open-air: the level of the water inside also depends on the potential rain and evaporation.
It is possible to change the state of one valve only after one second since the last valve operation.
Figure~\ref{fig:TanksDescr} is a schematic view of the system.

The corresponding hybrid automaton has eight locations, one for each combination of the state (open/closed) of the three valves, and three variables: $x$ and $y$ for the water level in the tanks, and $t$ as
the clock that enforces a one-time-unit wait between consecutive discrete transitions.
Since the tanks are in the same geographic location, rain and evaporation are assumed to have the same
rate in both tanks, thus leading to a proper LHA, that is not rectangular~\cite{hybridgames99}.

\begin{figure}[h!]
\centering
\subfigure[Schema of the system.]{
\label{fig:TanksDescr}
\includegraphics[width=1.8in]{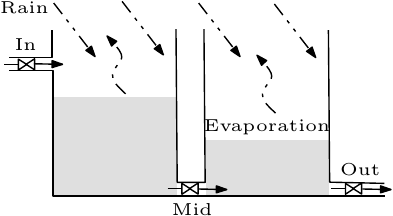}
}
\subfigure[Result for all valves closed and $t=0$.]{
\label{fig:TanksAllClose}
\includegraphics[width=1.3in]{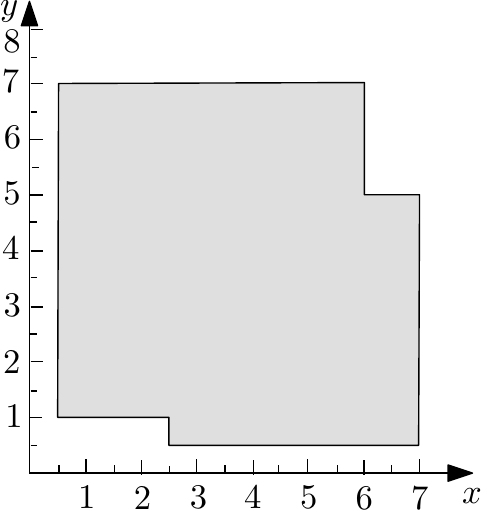}
}
\subfigure[Result with only \textit{mid} valve closed and $t=0$.]{
\label{fig:TanksMidClose}
\includegraphics[width=1.3in]{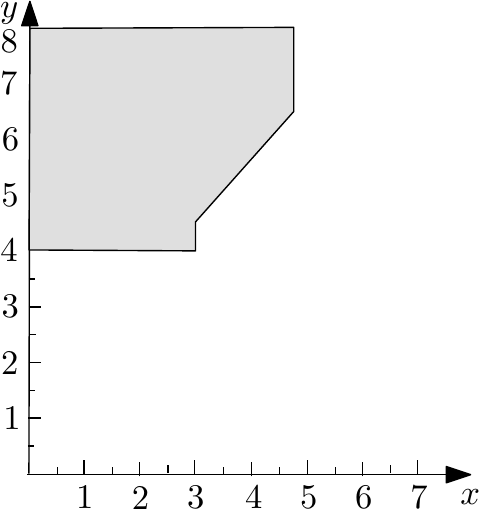}
}
\caption{Water Tank Control example.}
\label{fig:tanks}
\end{figure}

We set the \textit{in} and \textit{mid} flow rate to $1$, the \textit{out} flow rate to $3$, the maximum evaporation rate to $0.5$ and maximum rain rate to $1$, and solve the synthesis problem for the safety specification requiring the water levels to be between $0$ and $8$.
Figure \ref{fig:TanksAllClose} (resp., \ref{fig:TanksMidClose}) shows the fixpoint result in the case of all valves closed (resp., \textit{in} and \textit{out} open and \textit{mid} closed).
Due to the necessity of one second wait before taking a discrete action, in the case of Figure \ref{fig:TanksAllClose}, $x$ and $y$ must be between $0.5$ and $7$: otherwise, for example with a level greater than $7$ and maximum rain, after one second the level will exceed the limit. In a similar way, with a level less than $0.5$ and maximum evaporation, after one second the level would go below the lower bound.
The result is computed after 5 iterations in 11 seconds.

\section{Conclusions} \label{sec:conclusions}

We revisited the problem of automatically synthesizing a switching
controller for an LHA w.r.t.\ safety objectives.  The synthesis
procedure is based on the $\rwamay$ operator, for which we presented a
novel fixpoint characterization and formally proved its termination.

To the best of our knowledge, this represents the first sound and
complete procedure for the task in the literature.
%
%which is sound, complete, and effective.
We extended the tool PHAVer with our synthesis procedure
and performed a series of promising experiments.
An account of the challenges involved in the implementation
is presented in~\cite{gandalf11}.

%We are currently working on the automatic construction
%of a concrete control strategy, which, coupled with the hybrid
%system, would result in a closed system, amenable to automatic
%verification by state-of-the-art analysis tools.
%Moreover, we are investigating the synthesis problem w.r.t.
%reachability objectives.

%\bibliographystyle{plain}
%\bibliography{../hybrid}

\end{document}